\documentclass{lmcs} %
\pdfoutput=1

\usepackage{lastpage}
\lmcsdoi{22}{4}{24}
\lmcsheading{}{\pageref{LastPage}}{}{}%
{Dec.~23,~2025}{Dec.~01,~2026}{}


\usepackage[utf8]{inputenc}

\usepackage{hyperref}
\usepackage{amsmath}
\usepackage{amssymb}

 \newcommand{\pref}{\sqsubseteq_\mathsf{pref}}
 \newcommand{\suff}{\sqsubseteq_\mathsf{suff}}
 \newcommand{\spref}{\sqsubset_\mathsf{pref}}
 \newcommand{\ssuff}{\sqsubset_\mathsf{suff}}

\newcommand{\proot}{\varrho}
  
\newcommand{\vars}{X}

\newcommand{\qrank}{\mathsf{qr}}

\mathchardef\mhyphen="2D

\newcommand{\fraisse}{Fra\"{\i}ss\'e}

\newcommand{\ef}{Ehrenfeucht-\fraisse \xspace}

\newcommand{\concrel}{R_\circ}

\newcommand{\qr}{\mathsf{qr}}

\newcommand{\facts}{\mathsf{facts}}

\newcommand{\subs}{\sigma}

\newcommand{\mv}{\mathsf{w}}

\newcommand{\fc}{\mathsf{FC}}

\newcommand{\logeq}{\mathbin{\dot{=}}}

\newcommand{\lang}{\mathcal{L}}

\newcommand{\signature}{\tau}

\newcommand{\emptyword}{\varepsilon}

\newcommand{\union}{\mathrel{\cup}}
\newcommand{\intersect}{\mathrel{\cap}}

\newcommand{\df}{:=}

\renewcommand{\epsilon}{\varepsilon}

\newcommand{\rnc}[1]{\renewcommand{#1}}

\rnc{\leq}{\ensuremath{\leqslant}}
\rnc{\geq}{\ensuremath{\geqslant}}

\rnc{\le}{\leq}
\rnc{\ge}{\geq}

\newcommand{\isdef}{\ensuremath{:=}}
\newcommand{\deff}{\isdef}

\newcommand{\set}[1]{\ensuremath{\{#1\}}}

\newcommand{\setc}[2]{\set{#1  \mid  #2}}

\newcommand{\NN}{\ensuremath{\mathbb{N}}}
\newcommand{\NNpos}{\ensuremath{\mathbb{N}_+}}

\newcommand{\Structure}[1]{\ensuremath{\mathcal{#1}}}

\newcommand{\Aut}{\ensuremath{\mathcal{M}}}

\usepackage{stmaryrd}
\usepackage{xspace}
\usepackage{tikz}
\usetikzlibrary{arrows,trees,positioning,calc}
\usetikzlibrary{automata,shapes.arrows,backgrounds,fit}
\usetikzlibrary{decorations.pathmorphing,decorations.pathreplacing}
\usetikzlibrary{shapes}
\tikzstyle{rstate}=[state,ellipse]
\tikzset{>={latex}}

\usepackage{algorithm}
\usepackage{algpseudocode}
\theoremstyle{plain} %

\specialissue{Selected Papers of the 39th and 40th ACM/IEEE Symposium on Logic in Computer Science (LICS 2024 and 2025)}

\begin{document}

\title[Characterization and Decidability of FC-Definable Regular Languages]{Characterization and Decidability of FC-Definable Regular Languages}
\thanks{The majority of this work was completed while the first author was at Loughborough University.}	%

\author[S.~M.~Thompson]{Sam M. Thompson\lmcsorcid{0000-0002-3476-6739}}[a]
\author[N.~Schweikardt]{Nicole Schweikardt\lmcsorcid{0000-0001-5705-1675}}[b]
\author[D.~D.~Freydenberger]{Dominik D. Freydenberger\lmcsorcid{0000-0001-5088-0067}}[c]

\address{University of Liverpool, Liverpool, UK}
\email{S.M.Thompson2@liverpool.ac.uk}
\address{Humboldt-Universität zu Berlin, Berlin, Germany}
\email{schweikn@informatik.hu-berlin.de}
\address{Loughborough University, Loughborough, UK}	%
\email{d.d.freydenberger@lboro.ac.uk}

\begin{abstract}
  \noindent $\fc$ is a first-order logic that reasons over all factors of a finite word using concatenation, and can define non-regular languages like that of all squares ($ww$).
In this paper, we establish that there are regular languages that are not $\fc$-definable.
Moreover, we give a decidable characterization of the $\fc$-definable regular languages in terms of algebra, automata, and regular expressions.
The latter of which is natural and concise: Star-free generalized regular expressions extended with the Kleene star of terminal words.
\end{abstract}

\maketitle
\section{Introduction}\label{sec:intro}
The logic $\fc$ was introduced by Freydenberger and Peterfreund~\cite{frey2019finite} as a new approach to first-order logic over finite words.
Commonly, logic treats words as a sequence of positions which are given symbols with symbol predicates (e.g., see~\cite{straubing2012finite}).
Instead, $\fc$ reasons over the set of factors of an input word using concatenation $x \logeq y \cdot z$ along with constant symbols for terminal symbols and for the empty word.
For example, $\exists x , y, z \colon (x \logeq y \cdot y) \land (y \logeq \mathtt{b} \cdot z)$, with $\mathtt{b}$ being a terminal symbol, defines the language of those words that contain a factor of the form $\mathtt{b} v \mathtt{b} v$ with $v \in \Sigma^*$. 
$\fc$ is a finite model variant of the \emph{theory of concatenation} (introduced by Quine~\cite{quine1946concatenation}).
Both combine \emph{word equations} (which have been extensively studied~\cite{day2024closer, karhumaki2000expressibility, karhumaki2001expressibility, lin2016string, plandowski2004satisfiability}) with first-order logic.

A motivation for $\fc$ is the connection to \emph{document spanners}, which were introduced by Fagin, Kimelfeld, Reiss and Vansummeren~\cite{fag:spa}.
Document spanners query text documents by first extracting tables using regular expressions, and then applying relational algebra operators to those tables.

$\fc$ extended with \emph{regular constraints} (atomic formulas which ensure that a variable is replaced with a word from a certain regular language) has the same expressive power as the \emph{generalized core spanners}~\cite{frey2019finite}, a particular class of document spanners.
Moreover, the existential-positive fragment of $\fc$ with regular constraints\footnote{Although, the existential-positive fragment requires an extra constant symbol, which Freydenberger and Peterfreund~\cite{frey2019finite} call the \emph{universe variable}.} has the same expressive power as the \emph{core spanners}~\cite{frey2019finite} -- which were introduced by Fagin et al.~\cite{fag:spa} to capture the core functionality of IBM's \emph{Annotation Query Language}.
This strong connection between $\fc$ and document spanners allows one to bring techniques from finite model theory to text querying, such as formulas with bounded width~\cite{frey2019finite}, acyclic conjunctive queries~\cite{freydenberger2021splitting}, and \ef games~\cite{thompson2023generalized}.
However, the use of regular constraints naturally raises the question as to whether they are needed.

This paper shows that the class of $\fc$-definable regular languages is a proper subset of the regular languages (over non-unary alphabets\footnote{The situation is more straightforward for unary alphabets: $\fc$ over a unary alphabet is exactly the \emph{semi-linear languages}~\cite{thompson2023generalized} and thus can define all the unary regular languages~\cite{parikh1966context}.}).
This demonstrates that an extension by regular constraints is necessary when using $\fc$ as a logic for document spanners.
Moreover, stepping aside from the database theory aspects of $\fc$, this paper provides a comprehensive answer to a fundamental question regarding the expressive power of first-order logic with concatenation.
It follows in the long-standing tradition of characterizing the regular languages definable in various logics (see Straubing~\cite{straubing2012finite} as a starting point), focusing on the logic $\fc$.  
The main results provide a decidable characterization of the $\fc$-definable regular languages in terms of (generalized) regular expressions, automata, and algebra.

\paragraph*{Regular Expression Characterization}
Arguably, the most natural and concise formulation of the $\fc$-definable regular languages is the regular expression formulation:
Take the star-free generalized regular expressions (which are usual regular expressions without Kleene star, but with complement) and add the Kleene star of terminal words. 
For example, $\bigl( \mathtt{abab} \cdot (\mathtt{abab})^* \bigr)^c$, where the superscript $c$ denotes complement, describes the language of words that are not of the form $(\mathtt{ab})^{2n}$ for any $n > 0$.
To formalize this class, we use the \emph{star-free closure} operator; introduced by Place and Zeitoun in~\cite{place2019all}.
However, we rely on the results and formulations given in the more recent version~\cite{place2023closing}.
The star-free closure takes a class of regular languages $\mathcal{C}$, and maps it to $\mathsf{SF}(\mathcal{C})$ which is the smallest class containing $\mathcal{C}$, the singletons $\{ \mathtt a \}$ for terminal symbols, and is closed under concatenation, union, and complement.
A main result of this paper shows
that the class of $\fc$-definable regular languages is exactly $\mathsf{SF}(\mathcal{R})$ where $\mathcal{R} \df \{ w^* \mid w \in \Sigma^* \}$.
Using the notion of star-free closure allows us to draw upon the results and techniques of~\cite{place2019all,place2023closing}, in particular, the algebraic characterization of the star-free closure of so-called \emph{prevarieties}.

\paragraph*{Automata Characterization}
The automata characterization is given by a necessary and sufficient
criterion relying on a new notion 
which we call a \emph{loop-step cycle}.
A minimal DFA has a loop-step cycle if there are two words $w$ and $v$, that are not repetitions of the same word, and a sequence of $n\geq 2$ pairwise distinct states $p_0, p_1, \dots, p_{n-1}$ such that reading $w$ in any state $p_i$ results in again being in state $p_i$, and  reading $v$ in any state $p_i$ result in moving to state $p_{i+1 \pmod{n}}$.  
Notice that the minimal DFA for the language of words with an even number of $\mathtt{a}$ symbols (over the alphabet $\{ \mathtt a, \mathtt b \}$) has a loop-step cycle (see~\autoref{fig:even-a}).
A main result of this paper shows that a regular language is $\fc$-definable if, and only if, its minimal DFA does \emph{not} have a loop-step cycle.
We shall use this automata characterization to show that it is decidable (and $\mathsf{PSPACE}$-complete) to determine whether a regular language (given as a minimal DFA) is $\fc$-definable.

\begin{figure}
    \centering
   \begin{tikzpicture}[>=stealth',shorten >=1pt,auto,node distance= 1cm, scale = 1, transform shape,initial text=, bend angle = 20]
    \tikzstyle{vertex}=[circle,fill=white!25,minimum size=12pt,inner sep=3pt,outer sep=0pt,draw=white]
        \node[state, initial, accepting] (p0)  at (0,0)                   {$p_0$};
        \node[state] (p1) at (2,0) {$p_1$};
        
        \path[->] (p0) edge [bend left]   node [align=center]  {$\mathtt{a}$} (p1)
        (p1) edge [bend left]   node [align=center]  {$\mathtt{a}$} (p0)
        (p0) edge [loop above] node [align=center] {$\mathtt b$} (p0)
        (p1) edge [loop above] node [align=center] {$\mathtt b$} (p1);
    \end{tikzpicture}
    \caption{The minimal DFA for the language of words with an even number of $\mathtt{a}$ symbols.}
    \label{fig:even-a}
\end{figure}

\paragraph*{Algebraic Characterization}
Every regular language $L$ is associated with a finite monoid $M_L$ known as the \emph{syntactic monoid}\footnote{For this introduction, the precise definition of $M_L$ is not important.}, and a function $\eta_L$ called the \emph{syntactic morphism} which maps words to elements of this monoid (i.e., $\eta_L \colon \Sigma^* \rightarrow M_L$), see Pin~\cite{pin2010mathematical} for example.
Then, we have that $L = \eta_L^{-1}(F)$ for some $F \subseteq M_L$.
A main result of this paper characterizes
the $\fc$-definable regular languages by a class of syntactic morphisms which we call \emph{group primitive morphisms}.
A syntactic morphism $\eta_L \colon \Sigma^* \rightarrow M_L$ is group primitive if the inverse of any \emph{periodic element} $m \in M_L$ -- that is, where $m^n \neq m^{n+1}$ for all $n \in \mathbb{N}_+$ -- is a subset of $w^*$ for some $w \in \Sigma^*$.

\paragraph*{Related Work}
We refer to the more common first-order logic over strings (a linear order with symbol predicates) by $\mathsf{FO}[<]$.
Our characterization of the $\fc$-definable regular languages parallels the characterization of $\mathsf{FO}[<]$-definable regular languages~\cite{mcnaughton1971counter, schutzenberger1965finite}.
Each of the formalisms of our characterization naturally generalizes the characterizations of the $\mathsf{FO}[<]$-languages:
The class $\mathsf{SF}(\mathcal{R})$ extends the star-free languages~\cite{schutzenberger1965finite}.
The notion of automata with loop-step cycles is analogous to \emph{finite-automaton cycle existence}~\cite{cho1991finite}.
Lastly, group primitive languages generalize the languages definable by \emph{aperiodic monoids}~\cite{schutzenberger1965finite}.

Regarding the expressive power of $\fc$, Freydenberger and Peterfreund~\cite{frey2019finite} showed that $\setc{\mathtt a^n \mathtt b^n}{n\geq 0}$ is not $\fc$-definable.
Then, using \ef games, Thompson and Freydenberger~\cite{thompson2023generalized} provided a general tool for $\fc$ inexpressibility called the \emph{fooling lemma}.
Conjunctive query fragments of $\fc$ were considered by Thompson and Freydenberger~\cite{thompson2024languages}, where the expressive power was compared to other language generators.
Further work on $\fc$ includes an $\fc$-variant of Datalog~\cite{BellDF25}, acyclic conjunctive queries~\cite{freydenberger2021splitting}, and connections to cyber security~\cite{bell2025parsing}.

\subsection*{Related Version}
The present paper is an extended version of~\cite{ThompsonSF25} which appeared at LICS 2025.
This version includes full proofs as well as a new result in~\autoref{sec:ac0}.

\subsection*{Structure of the Paper}
\autoref{sec:prelims} gives some preliminary definitions.
We start~\autoref{section:main-results} by giving the formal definitions of star-free closure (\autoref{sec:star-free-closure-defn}), group primitive languages (\autoref{sec:group-primitive-defn}), and loop-step cycles (\autoref{sec:loop-step-cycle-defn}); all of which are required for formulating our main result: a decidable characterization of the $\fc$-definable regular languages (\autoref{sec:main-results}).
At the end of~\autoref{sec:main-results}, we give an overview of the proof structure of this main result, while
Sections~\ref{sec:gp-lsc}--\ref{sec:fc-gp} are
dedicated to actually proving this characterization.
\autoref{sec:ac0} shows that the $\fc$-definable regular languages over a two-letter alphabet are a strict subclass of the regular languages in $\mathsf{AC}^0$.
We close the paper in \autoref{sec:conclusion}.

\section{Preliminaries}\label{sec:prelims}
Let $\mathbb{N} \df \{0,1,2,\dots\}$ and let $\mathbb{N}_+ \df \mathbb{N} \setminus \{ 0 \}$ where $\setminus$ denotes set difference.
The cardinality of a set $S$ is denoted by $|S|$.
For $n \in \NNpos$, we use $[n]$ for $\setc{i\in\NN}{1\leq i\leq n}$.
For a vector $\vec a \in A^k$, for some set $A$ and $k \in \mathbb{N}$, we write $x \in \vec a$ to denote that $x$ is a component of $\vec a$.
For a finite set $A \subset \mathbb{N}$, the minimum and maximum elements of $A$ are denoted by $\mathsf{min}(A)$ and $\mathsf{max}(A)$ respectively.

\paragraph*{Words and Languages}
We use $\Sigma$ for a fixed and finite alphabet of terminal symbols
of size
$|\Sigma|\geq 2$. We write $\Sigma^*$ for the set of all words of finite length built from symbols in $\Sigma$, we let $\emptyword$ denote the empty word, and we let $\Sigma^+\deff\Sigma^*\setminus\set{\emptyword}$.
For a word, $w\in\Sigma^*$ we write $w^*$ for the language $\setc{w^n}{n\in\NN}$, where $w^n$ is $n$ consecutive repetitions of $w$.

If $w = w_1 \cdot w_2 \cdot w_3$ where $w, w_1, w_2, w_3 \in \Sigma^*$, then $w_1$ is a \emph{prefix} of $w$ (denoted $w_1 \pref w$), $w_2$ is a \emph{factor} of $w$ (denoted $w_2 \sqsubseteq w$), and $w_3$ is a \emph{suffix} of $w$ (denoted $w_3 \suff w$).
If~$w_2 \neq w$ also holds, then $w_2 \sqsubset w$, and we use the analogous symbols $\ssuff$ and $\spref$.
If $w_1 \neq \emptyword \neq w_3$,
then we call $w_2$ an \emph{internal factor} of $w$.
The set of all factors of $w \in \Sigma^*$ is denoted by $\facts(w) \df \{ u \in \Sigma^* \mid u \sqsubseteq w \}$.
We use $|w|$ for the length of $w \in \Sigma^*$; and for some $\mathtt{a} \in \Sigma$, we use $|w|_\mathtt{a}$ to denote the number of occurrences of $\mathtt{a}$ within $w$.

A word $w \in \Sigma^+$ is \emph{primitive} if for any $u \in \Sigma^+$ and $n \in \mathbb{N}$, we have that $w = u^n$ implies $n = 1$.
In other words, $w \in \Sigma^+$ is primitive if is not a repetition of a smaller word.
A word $w \in \Sigma^+$ is \emph{imprimitive} if it is not primitive.
For any word $w \in \Sigma^+$, the \emph{primitive root} of $w$ is the unique primitive word $\proot(w) \in \Sigma^+$ such that $w = \proot(w)^k$ for some $k \geq 1$.
For a language $L \subseteq \Sigma^*$, let~$\proot(L) \df \{ \proot(w) \mid w \in L \setminus \{ \emptyword \} \} $.

\paragraph*{Algebraic Concepts}
A \emph{monoid} $(M, \cdot, e)$ is a set $M$ with an associative multiplication operation $\cdot$ and an identity element $e$. 
When the multiplication and identity are clear from context, we shall denote a monoid simply by its set $M$.
Notice that $\Sigma^*$ with concatenation is a monoid with $\emptyword$ being the identity.
Given two monoids $M$ and $N$, a \emph{morphism} is a function $f \colon M \rightarrow N$ such that $f(x \cdot y) = f(x) \cdot f(y)$ for all $x,y\in M$.

The \emph{syntactic congruence} of $L \subseteq \Sigma^*$ is the relation $\sim_L$ defined as $u \sim_L v$ if and only if $xuy \in L \Leftrightarrow xvy \in L$ for all $x, y \in \Sigma^*$.
The set of equivalence classes in $\Sigma^*$ with respect to $\sim_L$ is denoted $\Sigma^* / {\sim_L}$. The morphism $\eta_L \colon \Sigma^* \rightarrow \Sigma^* / {\sim_L}$ which maps words to their equivalence class is called the \emph{syntactic morphism of $L$}.

Given a language $L \subseteq \Sigma^*$, the set $\Sigma^* / {\sim_L}$ with the multiplication $\eta_L(u) \cdot \eta_L(v) = \eta_L(u v)$ forms a monoid called the \emph{syntactic monoid of $L$}, which we often denote with $M_L$, see~\cite{pin2010mathematical} for more details.
We say that $x \in M_L$ is \emph{aperiodic} if there exists some $n \in \mathbb{N}_+$ such that $x^n = x^{n+1}$; otherwise,  $x$ is \emph{periodic}.
A monoid $M$ is aperiodic if every element is aperiodic (see Sch{\"u}tzenberger~\cite{schutzenberger1965finite}).
We note that for any finite monoid $M$, and any $x \in M$, there exist $j,p > 0$ such that
$x^{pn+j} = x^j$ for all $n \in \mathbb{N}$ (e.g., see Section~6, Chapter~2 of Pin~\cite{pin2010mathematical}).
We call an element $i \in M$ of a monoid \emph{idempotent} if $i = i \cdot i$. 

For a monoid $M$, a subset $G \subseteq M$ is a \emph{subgroup} if $G$ forms a group (with
the same multiplication as $M$, but with
a potentially different identity element than $M$), and if $|G| = 1$ then $G$ is a \emph{trivial subgroup}.

\paragraph*{The Logic FC}
We assume the reader is familiar with the standard concepts of first-order logic (for example, see Libkin~\cite{libkin2004elements}).
However, we shall look at the particular logic with which this paper is concerned in more detail.
The definitions given here are based on the definitions given by Thompson and Freydenberger in~\cite{thompson2023generalized} -- who define $\fc$ in a slightly more technical way than its original definition by Freydenberger and Peterfreund~\cite{frey2019finite}.

$\fc$ is built on one fixed signature $\signature_\Sigma \df \{ \concrel, \emptyword\}\cup \{ \mathtt{a} \mid \mathtt{a} \in \Sigma \}$ for every terminal alphabet~$\Sigma$, where $\concrel$ is a ternary relation symbol and where $\emptyword$ and each $\mathtt{a} \in \Sigma$ is a constant symbol.
Given a word $w \in \Sigma^*$, the $\signature_\Sigma$-structure that represents $w$ is defined as $\Structure{A}_w \df (A, \concrel^{\Structure{A}_w}, \emptyword^{\Structure{A}_w}, (\mathtt{a}^{\Structure{A}_w})_{\mathtt a \in \Sigma} )$ where
\begin{itemize}
\item $A \df \facts(w) \union \{ \perp \}$ is the universe,
\item $\concrel^{\Structure{A}_w} \df \{ (a,b,c) \in \facts(w)^3 \mid a = b \cdot c \}$ 
\item $\mathtt{a}^{\Structure{A}_w} \df \mathtt{a}$ if $|w|_\mathtt{a} \geq 1$, and $\mathtt{a}^{\Structure{A}_w} = \perp$ otherwise, and
\item $\emptyword^{\Structure{A}_w} \df \emptyword$.
\end{itemize}
We call such a structure $\Structure{A}_w$ an \emph{$\fc$-structure}.

Note that if $|w|_\mathtt{a} = 0$, then $\mathtt{a}^{\Structure{A}_w} = \perp$.
However, we usually deal with those words $w$ where $|w|_\mathtt{a} \geq 1$ for all $\mathtt{a} \in \Sigma$.
Therefore, we tend to write $\mathtt{a} \in \Sigma$ rather than $\mathtt{a}^{\Structure{A}_w} \in A$.

Let $\vars$ be a fixed, countably infinite set of variables (disjoint from $\Sigma$ and $\signature_\Sigma$).
An $\fc$-formula is a first-order formula whose atomic formulas are $\concrel(x,y,z)$, where $x$, $y$, and $z$ are variables or constants.
As syntactic sugar, we write $(x \logeq y \cdot z)$ for atomic $\fc$ formulas, as we always interpret $\concrel$ as concatenation.
More formally:
\begin{defi}
Let $\fc$ be the set of all $\fc$-formulas defined recursively as:
\begin{itemize}
\item If $x, y, z \in \vars \union \Sigma \union \{ \emptyword \}$, then $(x \logeq y \cdot z) \in \fc$,
\item if $\varphi, \psi \in \fc$, then $(\varphi \land \psi), (\varphi \lor \psi), \neg \varphi \in \fc$, and
\item if $\varphi \in \fc$ and $x \in \vars$, then $\forall x \colon \varphi \in \fc$ and $\exists x \colon \varphi \in \fc$.
\end{itemize}
\end{defi}
\noindent
Parentheses are freely omitted when the meaning is clear.
We allow atomic formulas of the form $x \logeq y$, as this can be expressed by $x \logeq y \cdot \emptyword$. 
We also use $Q x_1, x_2, \dots, x_n \colon \varphi$ as shorthand for $Q x_1 \colon Q x_2 \colon \dots Q x_n \colon \varphi$ where $Q \in \{ \exists, \forall \}$.

In $\fc$, an interpretation $\mathcal{I} \df (\Structure{A}_w, \subs)$ consists of a $\signature_\Sigma$-structure~$\Structure{A}_w$ that represents some $w \in \Sigma^*$, and a mapping~$\subs \colon \vars \rightarrow \facts(w)$. 
Notice that $\subs(x) \neq \perp$ is assumed for all $x \in \vars$.
We write $\mathcal{I} \models \varphi$ to denote that $\varphi$ is true in $\mathcal{I}$, defined in the usual way (see Chapter 2 of~\cite{libkin2004elements} for example). 

If $\varphi$ is a sentence (that is, $\varphi$ has no free variables), then we simply write $\Structure{A}_w \models \varphi$.
Furthermore, as $\Structure{A}_w$ and $\Structure{A}_v$ are isomorphic if and only if $w = v$, 
we can use $w$ as a shorthand for $\Structure{A}_w$ when appropriate.

\begin{defi}
The language defined by a sentence $\varphi \in \fc$ is $\lang(\varphi) \df \{ w \in \Sigma^* \mid w \models \varphi \}$.
Let $\lang(\fc)$ be the class of languages definable by an $\fc$ sentence.
\end{defi}

\noindent
In contrast to $\mathsf{FO}[<]$, the logic $\fc$ can define non-regular languages.
Possibly the most straightforward example is 
$\exists \, \mv , x \colon \bigl( \textsc{Word}(\mv) \land (\mv \logeq x \cdot x) \bigr)$
where $\textsc{Word}(\mv) \df \forall x, y \colon \Bigl( \bigl( (x \logeq \mv \cdot y) \lor ( x \logeq y \cdot \mv) \bigr) \rightarrow (y \logeq \emptyword) \Bigr)$. 
That is, $\textsc{Word}(\mv)$ states $\mv$ is the whole word.
Now consider the formula
\[\varphi \df \exists\, \mv, x \colon \bigl( \textsc{Word}(\mv) \land (\mv \logeq x \cdot x) \bigr) \land \neg \exists y \colon (y \logeq \mathtt b). \]
Assuming $\Sigma = \{ \mathtt a, \mathtt b \}$, we have that $\lang(\varphi) = (\mathtt{aa})^*$ which is a regular language.
However, it cannot be expressed in $\mathsf{FO}[<]$ as it is not \emph{star-free}; e.g., see~\cite{pin2020prove}.

\section{Main Results}\label{section:main-results}
In this section, we
provide precise statements of this paper's main results: A decidable characterization of the $\fc$-definable regular languages,
including a regular expression-based characterization, an algebraic characterization, and an automata characterization.
Before formulating these main results, we first need to provide the necessary definitions.

\subsection{Star-Free Closure}\label{sec:star-free-closure-defn}
Familiarity of the basics regarding regular languages is assumed (for example, see Pin~\cite{pin2021handbook}).
The \emph{star-free closure} is an operator on classes of languages that is of particular interest to this article.
Place and Zeitoun~\cite{place2023closing} conducted a systematic study of this operator.

\begin{defi}[Star-free closure~\cite{place2023closing}]\label{defn:SF-C}
Let $\mathcal{C}$ be an arbitrary class of regular languages. 
The class $\mathsf{SF}(\mathcal{C})$ is defined as the smallest class of languages containing $\mathcal{C}$, singletons $\{ \mathtt a \}$ where $\mathtt a \in \Sigma$, and which is closed under union, complement, and concatenation.
\end{defi}

\noindent
Thus, the star-free closure naturally extends the star-free languages
(the regular languages definable in $\mathsf{FO}[<]$); see~\cite{place2023closing, schutzenberger1965finite, straubing2012finite} for more details.

This paper is particularly interested in the star-free closure of a particular class:
Let~$\mathcal{R}$ denote the class of languages of the form $\{ w^* \mid w \in \Sigma^* \}$.
Recalling~\autoref{defn:SF-C}, the class~$\mathsf{SF}(\mathcal{R})$ is recursively defined as follows:
\begin{itemize}
\item $\{ \mathtt a \} \in \mathsf{SF}(\mathcal{R})$ for all
  $\mathtt{a} \in \Sigma$ (we often write $\mathtt{a}$ for $\set{\mathtt{a}}$), and $w^* \in \mathsf{SF}(\mathcal{R})$ for any $w \in \Sigma^*$.
\item If $L, L' \in \mathsf{SF}(\mathcal{R})$, then each of the
  languages
  \ $L\union L'$, \ $\Sigma^* \setminus L$, \ $L \cdot L'$ \
  belong to $\mathsf{SF}(\mathcal{R})$.
\end{itemize}

\noindent
Note that $\mathsf{SF}(\mathcal{R})$ is closed under Boolean operations (union, intersection, complement/set difference).
Furthermore, notice that $\emptyset, \Sigma^*, \{ \emptyword \} \in
\mathsf{SF}(\mathcal{R})$ since $\emptyset = (\mathtt a \intersect
\mathtt b)$, $\Sigma^* = \Sigma^* \setminus \emptyset$, and $\{
\emptyword \} = \mathtt{a}^* \setminus \mathtt{a} \cdot \mathtt{a}^*$
for $\mathtt a , \mathtt{b}\in \Sigma$ with $\mathtt{a}\neq\mathtt{b}$.

\begin{exa}\label{example:sfr}
Let $\Sigma = \{ \mathtt{a}, \mathtt{b} \}$ and consider the regular language $L \df (\mathtt{aa} \union \mathtt{bb})^*$.
While it may not be immediately clear that $L$ belongs to $\mathsf{SF}(\mathcal{R})$, consider the following $\mathsf{SF}(\mathcal{R})$-language:
\[
L_1 \df  \bigl( (\{ \emptyword\} \union \Sigma^* \mathtt{b}) \cdot (\mathtt{aa})^* \mathtt{a} \cdot (\mathtt{b} \Sigma^* \union \{ \emptyword \} ) \bigr) \union 
\bigl( (\{ \emptyword\} \union \Sigma^* \mathtt{a}) \cdot (\mathtt{bb})^* \mathtt{b} \cdot (\mathtt{a} \Sigma^* \union \{ \emptyword \} ) \bigr) .
\]
For any $w \in L$ and any factor $\mathtt{a}^n \sqsubseteq w$ which is a prefix or preceded by $\mathtt{b}$, and is a suffix or succeeded by $\mathtt{b}$, it necessarily holds that $n$ is even.
The symmetric statement holds for the analogous factors $\mathtt{b}^n \sqsubseteq w$.
One can therefore verify that $\Sigma^* \setminus L_1 = L$, and thus $L$ does indeed belong to $\mathsf{SF}(\mathcal{R})$.
\end{exa}

\subsection{Group Primitive Languages}\label{sec:group-primitive-defn}

We now introduce a new class of regular languages defined by a restriction on the syntactic morphism:

\begin{defi}[Group Primitive Language]\label{defn:group-prim}
Let $M$ be a finite monoid.
A morphism $\mu \colon \Sigma^* \rightarrow M$ is \emph{group primitive} if~$| \proot(\mu^{-1}(x))| = 1$ for all periodic elements $x \in M$.
A language $L \subseteq \Sigma^*$ is \emph{group primitive} if it is regular and its syntactic morphism is group primitive. 
\end{defi}

\noindent
As an aside, let us note the connection to languages recognized by aperiodic monoids.
Recall that an aperiodic monoid is a monoid where all elements $x$ satisfy $x^n = x^{n+1}$ for some $n \in \mathbb{N}_+$.
This can alternatively be phrased as follows: An aperiodic monoid is a monoid where all subgroups are trivial (see~\cite{schutzenberger1965finite}, for example).
In comparison, while the syntactic monoid of a group primitive language may contain non-trivial subgroups, the preimage of any element of a non-trivial subgroup (under the syntactic morphism) is a language of the form $\{ w^i \mid i \in I \}$ for some $w \in \Sigma^*$ and $I \subseteq \mathbb{N}$.
This is because if $G \subseteq M_L$ is a subgroup, then the identity $e_G$ of $G$ is the only aperiodic element (if we take $x^{n} = x^{n+1}$ with $x \in G$, then a simple cancellation argument implies $x = e_G$).
Consequently, each $x \in G \setminus \{ e_G\}$ satisfies $|\proot(\mu^{-1}(x))| = 1$ and thus $\mu^{-1}(x) = \{ w^i \mid i \in I \}$ for some $w \in \Sigma^*$ and $I \subseteq \mathbb{N}$.
Moreover, this implies $\mu^{-1}(e_G) = \{ w^j \mid j \in J \}$ for some $J \subseteq \mathbb{N}$ by the following reasoning:
pick any $u \in \mu^{-1}(e_G)$ and any $w^i \in \mu^{-1}(x)$ and notice that $\mu(w^i \cdot u) = \mu(w^i) \cdot \mu(u) = x \cdot e_G = x$ and thus $w^i u = w^r$ since $w^i u \in \mu^{-1}(x)$.
Taking $w^i u = w^r$, we get $u = w^{r-i}$.

\begin{exa}
Let $\Sigma = \{ \mathtt a, \mathtt b \} $.
The language $L_1 \df (\mathtt{aa})^*$ has the syntactic monoid $\{ e, x, y \}$ with $e$ being the identity, $xx = e$, and $yx = xy =yy = y$.
The syntactic morphism $\eta_{L_1} \colon \Sigma^* \rightarrow \{ e, x, y \}$ is: $\eta_{L_1}(w) = e$ if $w \in (\mathtt{aa})^*$, $\eta_{L_1}(w) = x$ if $w \in \mathtt{a}(\mathtt{aa})^*$, and $\eta_{L_1}(w) = y$ if $|w|_\mathtt{b} \geq 1$.
Clearly, $x$ is the only periodic element, however as $\eta_{L_1}^{-1}(x) = \mathtt{a} (\mathtt{aa})^*$, we have that $L_1$ is a group primitive language since $\proot(\mathtt{a}(\mathtt{aa})^*) = \{ \mathtt a \}$. 

Now consider $L_2 \df \{w \in \Sigma^* \mid |w|_\mathtt{a} \text{ is even} \}$.
The syntactic monoid for $L_2$ is $\{ e, x \}$ with $e$ being the identity and $xx = e$.
The syntactic morphism $\eta_{L_2} \colon \Sigma^* \rightarrow \{e,x \}$ is $\eta_{L_2}(w) = e$ if $|w|_\mathtt{a}$ is even, and $\eta_{L_2}(w) = x$ otherwise.
Again, it is clear that $x$ is periodic, however, $\mathtt{ba}$ and $\mathtt{baaa}$ are in $\eta_{L_2}^{-1}(x)$ which means that $L_2$ is not group primitive as $\proot(\mathtt{ba}) \neq \proot(\mathtt{baaa})$.
\end{exa}

\subsection{Automata with Loop-Step Cycles}\label{sec:loop-step-cycle-defn}
A deterministic finite automaton (DFA) $\mathcal{M}$ is a tuple $(Q, \Sigma, \delta, q_0, F)$ with $Q$ being the set of states, $\Sigma$ being the alphabet, $\delta \colon Q \times \Sigma \rightarrow Q$ being the transition function, $q_0 \in Q$ being the start state, and $F \subseteq Q$ being the set of accepting states.
We write $\delta^* \colon Q \times \Sigma^* \rightarrow Q$ as the reflexive and transitive closure of $\delta$. 
Then, $\mathcal{M}$ defines the language $\lang(\mathcal{M})$ of all words $w \in \Sigma^*$ such that $\delta^*(q_0, w) \in F$.
We call a DFA \emph{minimal} if there does not exist an equivalent DFA with fewer states.

We now define a condition for minimal DFAs:

\begin{defi}[Loop-Step Cycle]\label{defn:loop-step}
Let $\mathcal{M} \df (Q, \Sigma, \delta, q_0, F)$ be a minimal DFA.
We say that $\mathcal{M}$ has a \emph{loop-step cycle} if there exist $n \geq 2$ pairwise distinct states $p_0, p_1, \dots, p_{n-1}$ and words $w,v \in \Sigma^+$, where $\proot(w) \neq \proot(v)$, such that:
\begin{itemize}
\item $\delta^*(p_i,w) = p_i$ for all $i \in \{0, \dots, n-1 \}$, and
\item $\delta^*(p_i, v) = p_{i+1}$ for all $i \in \{0, \dots, n-2 \}$ and $\delta^*(p_{n-1}, v) = p_0$.
\end{itemize}
\end{defi}

\noindent
We shall often write $\delta^*(p_i, v) = p_{i+1 \pmod{n}}$ to denote that $\delta^*(p_i, v) = p_{i+1}$ for all $i \in \{0, \dots, n-2 \}$ and that $\delta^*(p_{n-1}, v) = p_0$.

\begin{exa}\label{example:loop-step}
Consider the languages $L_1 \df (\mathtt{aa} \union \mathtt{ab} \union \mathtt{ba})^*$ and $L_2 \df (\mathtt{aa} \union \mathtt{ab} \union \mathtt{bb})^*$.
The minimal DFA for $L_1$ is given in~\autoref{fig:two-aut-lsp} (on the
left-hand side) and the minimal DFA for $L_2$ is given in~\autoref{fig:two-aut-lsp} (on the right-hand side).
Note that the automaton for $L_1$ has a loop-step cycle since $\delta(p_0, \mathtt{a}) = p_1$ and $\delta(p_1, \mathtt a) = p_0$, and $\delta^*(\mathtt{ba}, p_0) = p_0$ and $\delta^*(\mathtt{ba}, p_1) = p_1$.
Although it is somewhat tedious to verify, the automaton for $L_2$ does not have a loop-step cycle:
To see why this holds, consider the automaton $\mathcal M_2$ given
in~\autoref{fig:two-aut-lsp} for $L_2$, and assume that $\mathcal M_2$ has a loop-step cycle.
Hence, there are $n\geq 2$ pairwise distinct states $p_0, p_1, \dots, p_{n-1}$ and $w, v \in \Sigma^*$ where $\proot(w) \neq \proot(v)$  such that:
\begin{itemize}
\item $\delta^*(p_i,w) = p_i$ for all $i \in \{0, \dots, n-1 \}$, and
\item $\delta^*(p_i, v) = p_{i+1 \pmod{n}}$.
\end{itemize}

\noindent
Notice that $q_2$ cannot reach any other state, and therefore $q_2 \notin \{ p_0, \dots, p_{n-1} \}$.
Furthermore, the only incoming transition for $q_3$ is labelled
$\mathtt b$ and the only incoming transition for $q_1$ is labelled
$\mathtt{a}$ and thus $q_3$ and $q_1$ cannot both be in $\{ p_0, \dots, p_{n-1} \}$.

This leaves us with $\{ p_0, \dots, p_{n-1} \}$ being either $\{ q_0, q_1 \}$ or $\{ q_0, q_3 \}$.
For state $q_i$ and $q_j $, let $L_{q_i, q_j} \df \{ u  \mid \delta^*(q_i, u) = q_j \}$.
Then, one can verify that
\begin{itemize}
\item $L_{q_0, q_0} \intersect L_{q_1, q_1} = (\mathtt{aa})^*$ and $L_{q_0, q_1} \intersect L_{q_1, q_0} = \mathtt{a} (\mathtt{aa})^*$ and thus $\proot(w) = \proot(v)$. 
\item $L_{q_0, q_0} \intersect L_{q_3, q_3} = (\mathtt{bb})^*$ and $L_{q_0, q_3} \intersect L_{q_3, q_0} = \mathtt{b} (\mathtt{bb})^*$ and thus $\proot(w) = \proot(v)$.
\end{itemize}
Consequently, $\mathcal{M}_2$ cannot have a loop-step cycle. 
\begin{figure}
    \centering
    \begin{tikzpicture}[>=stealth',shorten >=1pt,auto,node distance=1.8 cm, scale = 1, transform shape,initial text=,bend angle=20]

    \node[initial,state,accepting] (A) {$p_0$};
    \node[state] (B) [right of=A] {$p_1$};
    \node[state] (C) [below of=B] {$p_2$};
    \node[state] (D) [left of=C] {$p_3$};

\path[->] (A) edge [bend left]   node [align=center]  {$\mathtt{a}$} (B)
(B) edge [bend left]   node [align=center]  {$\mathtt{a},\mathtt{b}$} (A)
      (A) edge [bend left]      node [align=center]  {$\mathtt{b}$} (D)
      (D) edge [bend left]      node [align=center]  {$\mathtt{a}$} (A)
      (C) edge [loop below] node [align=center]  {$\mathtt{a}, \mathtt{b}$} (C)
      (D) edge [below] node [align=center]  {$\mathtt{b}$} (C);

       \node[initial,state,accepting] (E) [node distance=4cm, right of=B] {$q_0$};
    \node[state] (F) [right of=E] {$q_1$};
    \node[state] (G) [below of=F] {$q_2$};
    \node[state] (H) [left of=G] {$q_3$};

\path[->] (E) edge [bend left]   node [align=center]  {$\mathtt{a}$} (F)
(F) edge [bend left]   node [align=center]  {$\mathtt{a},\mathtt{b}$} (E)
      (E) edge [bend left]      node [align=center]  {$\mathtt{b}$} (H)
      (H) edge [bend left]      node [align=center]  {$\mathtt{b}$} (E)
      (G) edge [loop below] node [align=center]  {$\mathtt{a}, \mathtt{b}$} (G)
      (H) edge [below] node [align=center]  {$\mathtt{a}$} (G);
    
    \end{tikzpicture}
    \caption{Minimal DFA for $(\mathtt{aa} \union \mathtt{ab} \union
      \mathtt{ba})^*$ on the left-hand side, and minimal DFA for $(\mathtt{aa} \union \mathtt{ab} \union \mathtt{bb})^*$ on the right-hand side. See~\autoref{example:loop-step}.}
    \label{fig:two-aut-lsp}
\end{figure}
\end{exa}

\noindent
We now show that it is decidable to determine whether a DFA has a loop-step cycle:

\begin{thm}\label{lemma:aut-PSPACE}
Deciding whether a minimal DFA has a loop-step cycle is $\mathsf{PSPACE}$-complete.
\end{thm}
\begin{proof}
Let $\mathcal{M} \df (Q, \Sigma, \delta, q_0, F)$ by a minimal DFA.
We shall prove that deciding whether $\mathcal{M}$ has a loop-step cycle is $\mathsf{PSPACE}$-complete.

An array $B$ of size $n$ is a \emph{cyclic shift} of an array $A$ of size $n$ if $B[i] = A[i+1]$ for all $0 \leq i < n-1$, and $B[n-1] = A[0]$.

\paragraph*{Upper Bound}
Upon input of a minimal DFA $\mathcal{M}$, we want to decide in $\mathsf{PSPACE}$ whether
$\mathcal{M}$ has a loop-step cycle.
By Savitch's Theorem (see Section 4.3 of~\cite{arora2009computational}), we have $\mathsf{PSPACE} = \mathsf{NPSPACE}$. This allows us to make non-deterministic guesses.
Clearly, \autoref{alg:PSPACE} runs non-deterministically with polynomial
space; hence it decides a problem that belongs to $\mathsf{PSPACE}$.
It remains for us to show the correctness of~\autoref{alg:PSPACE}.

\begin{algorithm}
    \caption{$\mathsf{NPSPACE}$ algorithm for the following: \\
    \textbf{Input}: Minimal DFA $\mathcal{M} = (Q, \Sigma, \delta, q_0, F)$. \\ 
    \textbf{Output}: $\mathsf{True}$ iff $\mathcal{M}$ has a loop-step cycle.}
    \label{alg:PSPACE}
    \begin{algorithmic}[1]
      \State Guess an integer $n$ with $2 \leq n \leq |Q|$
      \State Guess $n$ distinct states $p_0, \dots, p_{n-1} \in Q$
        \State Let $A$ be an array of size $n$ such that $A[i]  \leftarrow p_i$ for $0 \leq i < n$
        \State Let $B$ and $C$ be arrays of size $n$ such that $A = B = C$
        \State Let $\mathsf{diff} \leftarrow \mathsf{False}$
        \While{($B \neq C$) or ($B$ is not a cyclic shift of $A$) or ($\mathsf{diff} = \mathsf{False}$)}
        \State Guess $\mathtt{a}, \mathtt{b} \in \Sigma$
        \If{$\mathtt a \neq \mathtt b$}
        \State $\mathsf{diff} \leftarrow \mathsf{True}$
        \EndIf
        \State $B[i] \leftarrow \delta(B[i], \mathtt a)$ for each $0 \leq i < n$
        \State $C[i] \leftarrow \delta(C[i], \mathtt b)$ for each $0 \leq i < n$
        \EndWhile
        \State Return $\mathsf{True}$ 
    \end{algorithmic}
\end{algorithm}

Assume~\autoref{alg:PSPACE} returns true.
Let $u = \mathtt a_1 \cdots \mathtt a_m$ be the nondeterministic guesses for $B$, and let $v = \mathtt b_1 \cdots \mathtt b_m$ be the guesses for $C$.
Since $\mathsf{diff}$ is true, we know that $u \neq v$.
Furthermore, as $|u| = |v|$, we also know that $\proot(u) \neq \proot(v)$.
Since $B$ and $C$ are cyclic shifts of $A$, it follows that $\delta^*(p_i, u) = \delta^*(p_i, v) = p_{i + 1 \pmod{n}}$.
Hence, we have that $\delta^*(p_i, u) = p_{i+1 \pmod{n}}$ and for all $0 \leq i < n$ we have $\delta^*(p_i, v^n) = p_i$ where $\proot(u) \neq \proot(v^n)$ as $\proot(u) \neq \proot(v)$.
This concludes one direction.

For the other direction, let $\mathcal{M} = (Q, \Sigma, \delta, q_0, F)$ be a minimal DFA that has a loop-step cycle.
Let $w,v \in \Sigma^+$, where $\proot(w) \neq \proot(v)$, and let
$p_0,\ldots,p_{n-1}$ be $n$ distinct states with:
\begin{itemize}
\item $\delta^*(p_i,w) = p_i$ for all $i \in \{0, \dots, n-1 \}$, and
\item $\delta^*(p_i, v) = p_{i+1 \pmod{n}}$.
\end{itemize}
Let us assume~\autoref{alg:PSPACE} initializes $A$ such that $A[i] = p_i$ for $0 \leq i < n$.
Set $u_1 \df w^{n  |v|} \cdot v^{n  |w| + n+1}$ and $u_2 \df w^{2n|v|} \cdot v^{n+1}$. 
First, we show that $u_1 \neq u_2$:
\begin{align*}
w^{n  |v|} \cdot v^{n  |w| + n+1} & = w^{2n|v|} \cdot v^{n+1} \\
\Longrightarrow v^{n|w|} & = w^{n|v|}
\end{align*}
which contradicts $\proot(w) \neq \proot(v)$ as $n|w|\neq 0 \neq n|v|$. 
Next, notice that $|u_1| = |u_2|$:
\begin{itemize}
\item $|u_1| = n \cdot |w| |v|   + |v| ( n |w| + n + 1) = 2  n  |v|  |w| + n  |v| + |v|$,
\item $|u_2| = 2  n  \cdot |w| |v|   + |v| ( n + 1) = 2  n  |v| |w| + n |v| + |v|$. 
\end{itemize}
\noindent
If $\mathcal{M}$ is in state $p_i$ for some $0 \leq i < n$, then reading $w^k$ for any $k \in \mathbb{N}$ does not change the state of $\mathcal{M}$.
Likewise, if $\mathcal{M}$ is in state $p_i$, then reading $v^{kn+1}$ for any $k \in \mathbb{N}_+$ puts $\mathcal{M}$ into state $p_{i + 1 \pmod{n}}$.
Hence, for any $i \in [n]$ we have that $\delta^*(p_i, u_1) = \delta^*(p_i, u_2) = p_{i +1 \pmod{n}}$.

Assume that for $j$ with $1 \leq j \leq |u_1|$, the $j$-th time line 11 of~\autoref{alg:PSPACE} is executed, $\mathtt{a}$ is the $j$-th letter of $u_1$.
Likewise, assume that the $j$-th time line 12 of~\autoref{alg:PSPACE} is executed, $\mathtt{b}$ is the $j$-th letter of $u_2$.
Firstly, $\mathsf{diff} = \mathsf{True}$ since $u_1 \neq u_2$.
Moreover, since $\delta^*(p_i, u_1) = \delta^*(p_i, u_2) = p_{i +1 \pmod{n}}$ we have that $B=C$ and $B$ is a cyclic shift of $A$.

Thus, \autoref{alg:PSPACE} returns $\mathsf{True}$.

\paragraph*{Lower Bound}
To show $\mathsf{PSPACE}$-hardness, we reduce from the
$\mathsf{PSPACE}$-complete problem \emph{finite-automaton cycle
  existence}~\cite{cho1991finite} which is defined as follows: Given a
minimal DFA $\mathcal{M}$, decide whether $\mathcal{M}$ \emph{has a cycle},
i.e., whether there exists a word $u \in
\Sigma^*$ and a state $p$ of $\mathcal{M}$ such that $\delta^*(p, u)
\neq p$ and $\delta^*(p, u^r) = p$ for some $r \in \mathbb{N}_+$.

For an instance $\mathcal{M} = (Q, \Sigma, \delta, q_0, F)$ of the finite-automaton cycle existence problem, we construct the automaton $\mathcal{M}_2 = (Q_2, \Sigma_2, \delta_2, q_{0,2}, F_2)$ where 
\begin{itemize}
\item $Q_2 \df Q$,
\item $\Sigma_2 \df \Sigma \union \{ \bar{\mathtt{a}} \}$, where
  $\bar{\mathtt{a}}$ is a new letter with $\bar{\mathtt{a}}\not\in\Sigma$,
\item for all $q \in Q$, we have $\delta_2(q, \bar{\mathtt{a}}) = q$, and for any $q \in Q$ and $\mathtt a \in \Sigma$, we have $\delta_2(q, \mathtt a) \df \delta(q, \mathtt a)$,
\item $q_{0,2} \df q_0$, and
\item $F_2 \df F$.
\end{itemize}
Since $\mathcal{M}$ is a minimal DFA, $\mathcal{M}_2$ is also a
minimal DFA.
Clearly, constructing $\mathcal{M}_2$ from $\mathcal{M}$ can be done
in polynomial time with respect to the size of $\mathcal{M}$.
We now show that $\mathcal{M}_2$ has a loop-step cycle $\iff$
$\mathcal{M}$ has a cycle.

For direction ``$\Longleftarrow$'' assume that $\mathcal{M}$ has a
cycle, that is,
there exists some $u \in \Sigma^*$ and some state $q \in Q$ such that
$\delta^*(q, u) \neq q$ and $\delta^*(q,u^r) = q$ for some $r \in
\mathbb{N}_+$. W.l.o.g.\ assume that $r$ is chosen as small as
possible, and note that $r\geq 2$.
Then, there exist $r$ distinct states $p_0, \dots, p_{r-1}$ such that $\delta^*(p_i, u) = p_{i+1 \pmod{r}}$ for $0 \leq i < r$.
Trivially, $\delta_2^*(p_i, \bar{\mathtt a}) = p_i$ for all $0 \leq i < r$ and $\proot(u) \neq \proot(\bar{\mathtt a})$.
Hence, $\mathcal{M}_2$ has a loop-step cycle.

For direction ``$\Longrightarrow$'' we prove the contrapositive.
Assume $\mathcal{M}_2$ has a loop-step cycle.
Immediately, we know that there exists a word $u \in \Sigma_2^*$ and a state 
$p \in Q_2$ such that $\delta_2^*(p,u) \neq p$ and $\delta^*(p,u^r) = p$ for some $r \in \mathbb{N}_+$.
We may assume that $\bar{\mathtt{a}}$ does not occur in $u$ since $\delta_2(q, \bar{\mathtt{a}}) = q$ for all $q \in Q_2$.
This means that, in $\mathcal{M}$, we have that $\delta^*(p,u) \neq p$ and $\delta^*(p,u^r) = p$ for some $p \in Q$, $u \in \Sigma^*$, and $r \in \mathbb{N}_+$.
Thus, $\mathcal{M}$ has a cycle and consequently the original implication holds.

In summary, we have provided a polynomial-time reduction from the
$\mathsf{PSPACE}$-complete problem  finite-automaton cycle existence
to the problem of deciding whether a given minimal DFA has a
loop-step cycle. This shows that the latter problem is $\mathsf{PSPACE}$-hard.
\end{proof}

\subsection{Main Theorems}\label{sec:main-results}\label{sec:main-theorems}
With the prerequisites out of the way, we are now ready to state our main result:

\begin{thm}\label{thm:main}
Let $L\subseteq\Sigma^*$ be a regular language,  
and let  $\Aut \df (Q, \Sigma, \delta, q_0, F)$ be a minimal DFA with $L=\lang(\Aut)$.
Then, the following are equivalent:
\begin{enumerate}
\item\label{item:mainthm:fc} $L \in \lang(\fc)$,
\item\label{item:mainthm:sf} $L \in \mathsf{SF}(\mathcal{R})$,
\item\label{item:mainthm:gp} $L$ is group primitive,
\item\label{item:mainthm:lsc} $\mathcal{M}$ does not have a loop-step cycle.
\end{enumerate}
\end{thm}

\noindent
This provides a characterization of the regular languages that can be expressed in~$\fc$.
Our second main result is that this characterization is decidable, in fact (for minimal automata) it is $\mathsf{PSPACE}$-complete:

\begin{thm}\label{thm:PSPACE}
Given a minimal DFA $\mathcal{M}$, deciding whether $\lang(\mathcal{M}) \in \lang(\fc)$ is $\mathsf{PSPACE}$-complete.
\end{thm}
\noindent
\autoref{thm:PSPACE} is immediately obtained by combining~\autoref{thm:main} with~\autoref{lemma:aut-PSPACE}.

With the main characterization given (\autoref{thm:main}) and the result
that this characterization is decidable (\autoref{thm:PSPACE}),
Sections~\ref{sec:gp-lsc}--\ref{sec:fc-gp} of this article are
devoted to proving~\autoref{thm:main}:
\begin{itemize}
\item $\ref{item:mainthm:lsc}\Rightarrow\ref{item:mainthm:gp}$ is
shown in \autoref{sec:gp-lsc} (\autoref{thm:group-prim}); the proof is by a
direct construction.

\item $\ref{item:mainthm:gp}\Rightarrow\ref{item:mainthm:sf}$ is
shown in \autoref{sec:gp-sf} (\autoref{thm:gp-sf}).
To prove this, we use the algebraic characterization of the star-free
closure of so-called \emph{prevarieties} given by Place and
Zeitoun~\cite{place2023closing}, and show that if a language is group
primitive, then it satisfies
Place and Zeitoun's characterization.

\item $\ref{item:mainthm:sf}\Rightarrow\ref{item:mainthm:fc}$ is
shown in \autoref{sec:sf-fc} (\autoref{thm:sf-fc});
the proof is by a structural induction very similar to proofs of previous results (Lemma~5.5 of~\cite{frey2019finite} and Lemma~5.3 of~\cite{thompson2023generalized}).

\item $\ref{item:mainthm:fc}\Rightarrow\ref{item:mainthm:lsc}$ is
shown in \autoref{sec:fc-gp} (\autoref{thm:fc-lsc}).
The proof is quite involved; it relies on \ef
games, a result by Lynch \cite{lynch1982sets}, and a suitable
translation between words with concatenation and integers with addition.
\end{itemize}

\subsection{$\mathsf{SF}(\mathcal{R})$ and Regular $\mathsf{AC}^0$-Languages}\label{sec:ac0}
Before providing the proof of \autoref{thm:main}, let us briefly give some more context on the expressive power of $\mathsf{SF}(\mathcal{R})$ by comparing it to the class of regular languages in $\mathsf{AC}^0$.
For this brief section, we will not define $\mathsf{AC}^0$ and instead refer the reader to standard texts, such as Chapter~6 of~\cite{arora2009computational}.

Let $\textsc{Length}(q) \df \{ w \in \{ \mathtt 0, \mathtt 1 \}^* \mid |w| \equiv 0 \pmod{q} \}$.
We rely on the following result by Barrington, Compton, Straubing and Th{\'{e}}rien~\cite{BarringtonCST92}:

\begin{thmC}[{\cite[Theorem~3]{BarringtonCST92}}]\label{thm:AC0-reg}
Let $L \subseteq \{ \mathtt 0, \mathtt 1 \}^*$. Then, the following are equivalent:
\begin{enumerate}
\item $L$ is regular and $L \in \mathsf{AC}^0$,
\item $L$ belongs to the smallest family of subsets of $\{ \mathtt 0, \mathtt 1 \}^*$ that contains $\{ \mathtt 0 \}$, $\{ \mathtt 1 \}$, the languages $\textsc{Length}(q)$ for all $q > 1$, and is closed under concatenation and Boolean operations.
\end{enumerate}
\end{thmC}

\noindent
Note that~\cite{BarringtonCST92} also provides an algebraic and a
logical characterization for the class of regular
languages in $\mathsf{AC}^0$, but we omit the details to avoid additional definitions.

\begin{thm}
For $\Sigma\df\{ \mathtt 0, \mathtt 1 \}$, 
the class $\mathsf{SF}(\mathcal{R})$ is a strict subclass of the class of regular languages in $\mathsf{AC}^0$.
\end{thm}
\begin{proof}
By \autoref{thm:AC0-reg}, $\textsc{Length}(2)$ is a regular language
in $\mathsf{AC}^0$.
The minimal DFA for the language $\textsc{Length}(2)$ is given in~\autoref{fig:even}.
It is easy to see that this automaton has a loop-step cycle, as $\delta^*(p_i, \mathtt{00}) = p_i$ for $i \in \{0,1\}$, and $\delta^*(p_0, \mathtt{1}) = p_1$ and $\delta^*(p_1, \mathtt{1}) = p_0$.
Therefore, by~\autoref{thm:main} we have that $\textsc{Length}(2)
\notin \mathsf{SF}(\mathcal{R})$.

The rest of this proof is dedicated to showing that every language $L \in \mathsf{SF}(\mathcal{R})$ belongs to~$\mathsf{AC}^0$.
We continue with a structural induction along the recursive definition of $\mathsf{SF}(\mathcal{R})$.

\begin{figure}
    \centering
   \begin{tikzpicture}[>=stealth',shorten >=1pt,auto,node distance= 1cm, scale = 1, transform shape,initial text=, bend angle = 20]
    \tikzstyle{vertex}=[circle,fill=white!25,minimum size=12pt,inner sep=3pt,outer sep=0pt,draw=white]
        \node[state, initial, accepting] (p0)  at (0,0)                   {$p_0$};
        \node[state] (p1) at (2,0) {$p_1$};
        
        \path[->] (p0) edge [bend left]   node [align=center]  {$\mathtt{0}, \mathtt{1}$} (p1)
        (p1) edge [bend left]   node [align=center]  {$\mathtt{0}, \mathtt{1}$} (p0);
    \end{tikzpicture}
    \caption{The minimal DFA for the language $\textsc{Length}(2)$.}
    \label{fig:even}
\end{figure}

\noindent
Note that since the regular languages in $\mathsf{AC}^0$ are closed under union, complement, and concatenation, we only need to consider the base cases for the recursive definition of $\mathsf{SF}(\mathcal{R})$.
Moreover, $\{ \mathtt 0\}$ and $\{ \mathtt 1 \}$ are trivially in $\mathsf{AC}^0$.
Thus, all that remains is to show $w^*$ is in $\mathsf{AC}^0$ for any $w \in \{ \mathtt 0, \mathtt 1 \}^*$.

Let $L = w^*$ for some $w \in \{ \mathtt 0, \mathtt 1\}^*$. 
Note that $v^*$ for any primitive word $v \in \{ \mathtt 0, \mathtt 1 \}^*$ is a star-free language and thus belongs to $\mathsf{AC}^0$, see Theorem 3.2 and Theorem 3.3 of~\cite{mcnaughton1971counter}.
Then, consider $K \df \proot(w)^* \intersect \textsc{Length}(r)$ where $r \geq 1$ such that $w = \proot(w)^r$. 
Invoking~\autoref{thm:AC0-reg}, $K$ is a regular language and belongs to $\mathsf{AC}^0$.
The proof is then completed by noticing that $K = (\proot(w)^r)^*=  w^*$.
\end{proof}

\section{On Loop-Step Cycles and Group Primitive Regular Languages}\label{sec:gp-lsc}
In this section we prove the following theorem.

\begin{thm}\label{thm:gp-lsc}\label{thm:group-prim}
  Let $L$ be a regular language and let $\Aut$ be a minimal DFA with $L=\lang(\Aut)$.
  If $L$ is not group primitive, then $\Aut$ has a loop-step cycle.
\end{thm}
\begin{proof}
Let $L$ be a regular language that is not group primitive.
Let $\mathcal{M} \df (Q, \Sigma, \delta, q_0, F)$ be a minimal DFA with $L=\lang(\Aut)$.
Let $\eta_L \colon \Sigma^* \rightarrow M_L$ be the syntactic morphism of $L$, where $M_L$ is the syntactic monoid for $L$.
Since $L$ is not group primitive, $\eta_L$ is not group primitive.
That is, there exists a periodic element 
$g \in M_L$ where
$|\proot(\eta_L^{-1}(g))|\geq 2$.
Hence, there exist words
$w,v \in \eta_L^{-1}(g)$ such that $\proot(w) \neq \proot(v)$.

Since $g$ is periodic, we
have that $g^n \neq g^{n+1}$ for any $n \in \mathbb{N}$.
However, there exists $y \geq 0$ and $m > 0$ such that for all $x \in \NNpos$, we get $g^{y+ mx } = g^{y + m}$;
for example, see Section 6, Chapter 2 of~\cite{pin2010mathematical}.
We choose $m$ as small as possible.
Note that $m\geq 2$ because $g$ is periodic.
We let
$G\deff \set{ g^{y+m}, g^{y+m+1}, \dots, g^{y + 2m - 1}}$.
Observe that $G$ has $m$ elements, and $G$ is closed under multiplication with $g$, i.e., $g\cdot h=h\cdot g\in G$ for every $h\in G$.

For any words $u, u' \in \Sigma^*$ we have:
$\eta_L(u)=\eta_L(u')$ if and only if
$\delta^*(q, u) = \delta^*(q, u')$ for all $q \in Q$; see Proposition 4.28 of~\cite{pin2010mathematical}.
By our choice of $w$ and $v$ we have: $w,v\in\Sigma^+$, $\proot(w)\neq\proot(v)$, and $\eta_L(w)=\eta_L(v)=g$.

In particular, due to the fact that $m$ is chosen to be as small as possible, for any $z,z'\in\NN$ with $z,z'\geq y+m$ we have: 
\[
v^{z}\sim_L v^{z'} \iff  \eta_L(v^z)=\eta_L(v^{z'}) \iff  g^z=g^{z'} 
\iff  z -y \equiv z'-y \pmod{m}.
\]

Thus, since $m\not\equiv m{+}1 \pmod{m}$, we have $v^{m+y} \not\sim_L v^{m+y+1}$. Hence, there exists some state $p \in Q$ such that $\delta^*(p, v^{m+y}) \neq \delta^*(p, v^{m+y+1})$.
Let $p_i\deff \delta^*(p,v^{m+y+i})$ for every $i\in\set{0,\ldots,m{-}1}$.
Note that $p_0\neq p_1$ (because $p_0=\delta^*(p, v^{m+y})\neq \delta^*(p, v^{m+y+1}) =p_1$), and that $\delta^*(p_i,v)=p_{i+1}$ holds for every $i\in\set{0,\ldots,m-2}$.

Furthermore, $\delta^*(p_{m-1},v)=p_0$ due to the following:
Clearly, $\delta^*(p_{m-1},v)=\delta^*(p,v^{y+2m})$. Furthermore, $v^{y+2m}\sim_L v^{y+m}$ since
$\eta_L(v^{y+2m})= g^{y+2m}=g^{y+m}=\eta_L(v^{y+m})$. 
Thus, for \emph{every} state $q\in Q$ we have
$\delta^*(q,v^{y+2m})=\delta^*(q,v^{y+m})$. In particular for $q=p$ this yields
$\delta^*(p,v^{y+2m})=\delta^*(p,v^{y+m})=p_0$.

In summary, we now know that $\delta^*(p_i, v) = p_{i+1 \pmod{m}}$.
But 
note that the states $p_0, p_1, \dots, p_{m-1}$ are not necessarily pairwise distinct.
However, we know that, for all $x\in\NNpos$, we have that $v^{m+y} \sim_L v^{mx+y}$.  
Therefore, we get a set of states $\{ q_0, q_1, \dots, q_{m'-1} \}$ of size $m'$, where $2 \leq m' \leq m$, such that $q_0 = p_0$, $q_1 = p_1$ and $\delta^*(q_i, v) = q_{i+1 \pmod{m'}}$.
This is because we know that we eventually need to ``loop back'' to state $p_0$, which prohibits $\delta^*(p_i, v) = p_j$ where $p_j \neq p_0 \neq p_i$, and thus we obtain our set $\{ q_0, q_1, \dots, q_{m'-1} \}$.

Since $v^{m+y} \sim_L v^{mx+y}$ for all $x \in \NNpos$, it follows that $\delta^*(q, v^{m+y}) = \delta^*(q, v^{mx+y})$ for all $x \in \NNpos$ and all $q \in Q$.
Furthermore, we have that $w,v \in \eta_L^{-1}(g)$, thus $w^z \sim_L v^z$ for all $z \in \mathbb{N}$.
In particular, this implies for all $i\in\set{0,\ldots,m'{-}1}$ and the state $q_i$ that $\delta^*(q_i, w^{m+y}) = q_i$ for all $0 \leq i < m'$.
We let $w'\deff w^{m+y}$ and obtain that
$\delta^*(q_i, w') = q_i$ for all $i\in\set{0,\ldots,m'{-}1}$.
Furthermore, $\proot(w')=\proot(w) \neq \proot(v)$.

Consequently, we now know that $q_0,\ldots,q_{m'-1}$ are $m'\geq 2$ pairwise distinct states, $w'$ and $v$ are words with $\proot(w')\neq\proot(v)$, and $\delta^*(q_i,w')=q_i$ and $\delta^*(q_i,v)=q_{i+1 \pmod{m}}$ for all $i$ where $0 \leq i < m' {-} 1$.  
This proves that $\Aut$ has a loop-step cycle and therefore completes the proof of \autoref{thm:group-prim}.
\end{proof}

\noindent
The contraposition of~\autoref{thm:gp-lsc} shows the step $\ref{item:mainthm:lsc}\Rightarrow\ref{item:mainthm:gp}$ of~\autoref{thm:main}.

\section{Every $\mathsf{SF}(\mathcal{R})$-Language is Expressible in $\fc$}\label{sec:sf-fc}
In this section, we show that the star-free closure of the class $\{ w^* \mid w \in \Sigma^* \}$ can be expressed in $\fc$.
We know from Example 3.7 of~\cite{frey2019finite} that $\fc$ can express all the star-free languages, and moreover, Lemma 5.5 of~\cite{frey2019finite} shows that $w^*$ for any $w \in \Sigma^*$ can be expressed in $\fc$.
It is therefore unsurprising that we can combine these results to show $\mathsf{SF}(\mathcal{R}) \subseteq \lang(\fc)$.

A language $L \subseteq \Sigma^*$ is \emph{bounded} if $L \subseteq w_1^* \cdots w_n^*$ for some $w_1, \dots, w_n \in \Sigma^+$ and $n \geq 1$.
A language is a \emph{bounded regular language} if it is both bounded and regular (Ginsburg and Spanier~\cite{ginsburg1966bounded}).
Lemma 5.3 of~\cite{thompson2023generalized} states that the Boolean combination of bounded regular languages can be expressed in $\fc$.
Going only slightly beyond that result gives us:

\begin{thm}\label{thm:sf-fc}
If $L \in \mathsf{SF}( \mathcal{R})$, then $L \in \lang(\fc)$.
\end{thm}
\begin{proof}
First suppose that  for any $L \in \mathsf{SF}(\mathcal{R})$, we have $\psi_L(x)$ such that $(\Structure A_w, \subs) \models \psi_L(x)$ if and only if $\subs(x) \in L$.
Now let 
\[ \phi_L \df \exists x \colon \Bigl( \psi_L(x) \land \forall y, z \colon \bigl( ( (y \logeq x \cdot z) \lor (y \logeq z \cdot x) ) \rightarrow (z \logeq \emptyword) \bigr) \Bigr). \]
The formula $\phi_L$ states that there exists a factor $x$ such that $x \in L$, and for any factors $y, z$ with $y = xz$ or $y = zx$, it necessarily holds that $z = \emptyword$.
Consequently, $x$ represents the whole input word.
Thus, $\Structure A_w \models \phi_L$ if and only if $w \in L$.
The remainder of this proof establishes the following claim:
\begin{clm}
For every $L \in \mathsf{SF}(\mathcal{R})$, there exists an $\fc$-formula $\psi_L(x)$ such that $\subs(x) \in L$ if and only if $(\Structure A_w, \subs) \models \psi_L(x)$.
\end{clm}

\noindent
Let $L \in \mathsf{SF}(\mathcal{R})$.
We proceed by a structural induction (which follows very similarly to the proof of Lemma 5.5 of~\cite{frey2019finite} given in the version with proofs~\cite{frey2019finiteArx}).
\begin{itemize}
\item If $L = \{ \mathtt a \}$ for some $\mathtt a \in \Sigma$, then let $\psi_L(x) \df (x \logeq \mathtt{a} )$,
\item If $L = w^*$ for some $w \in \Sigma^*$, then let\footnote{We use $z^p$ as shorthand for $z$ concatenated with itself $p$-times; which can clearly be expressed in $\fc$.} \[ \psi_L(x) \df \exists y, z \colon \bigl( (x \logeq y \cdot \proot(w) ) \land (x \logeq \proot(w) \cdot y) \land (x \logeq z^p) \bigr), \]
where $p \in \mathbb{N}$ such that $w = \proot(w)^p$.
The correctness of $\psi_L(x)$ for this case follows from a result on commuting words: If $uv = vu$ with $u, v \in \Sigma^*$, then there exists $n,m \in \mathbb{N}$ and $w \in \Sigma^*$ with $u = w^n$ and $v = w^m$ (See Proposition 1.3.2 of Lothaire~\cite{lothaire1997combinatorics}, and see the proof of Lemma 5.5 of~\cite{frey2019finite} given in~\cite{frey2019finiteArx}).
\item If $L = L_1 \union L_2$, then we can assume the existence of $\psi_{L_1}(x)$ and $\psi_{L_2}(x)$ by the induction hypothesis. Then, let $\psi_L(x) \df \psi_{L_1}(x) \lor \psi_{L_2}(x)$.
\item If $L = \Sigma^* \setminus L_1$, then by the induction hypothesis, we assume the existence of $\psi_{L_1}$ and let $\psi_L(x) \df \neg \psi_{L_1}(x)$.
\item If $L = L_1 \cdot L_2$,  then we can assume the existence of $\psi_{L_1}(x_1)$ and $\psi_{L_2}(x_2)$ by the induction hypothesis, with $x$, $x_1$ and $x_2$ being pairwise distinct. Then, let 
\[ \psi_L(x) \df \exists x_1, x_2 \colon \bigl( (x \logeq x_1 \cdot x_2) \land \psi_{L_1}(x_1) \land \psi_{L_2}(x_2) \bigr). \]
\end{itemize}
This concludes the proof.
\end{proof}

\noindent
In this section, we have established $\ref{item:mainthm:sf}\Rightarrow\ref{item:mainthm:fc}$ of~\autoref{thm:main}.
By itself, \autoref{thm:sf-fc} may seem unsurprising.
However, the fact that $\fc$-definable regular languages are exactly $\mathsf{SF}(\mathcal{R})$ is rather surprising.
Especially considering the complicated languages $\fc$ can define (e.g., see Proposition 4.1 of~\cite{thompson2023generalized}).

\section{Every Group Primitive Language is in $\mathsf{SF}(\mathcal{R})$}\label{sec:gp-sf}
In this section, we show that any group primitive language (\autoref{defn:group-prim}) belongs to $\mathsf{SF}(\mathcal{R})$ (recall~\autoref{sec:star-free-closure-defn}).
To show this, we consider many concepts and results from Place and Zeitoun~\cite{place2023closing}; which algebraically characterize the star-free closure of so-called \emph{prevarieties}. 
We shall consider a prevariety $\mathcal{B}q$ (studied in Daviaud and Paperman~\cite{daviaud2018classes}) such that $\mathsf{SF}(\mathcal{B}q) = \mathsf{SF}(\mathcal{R})$, and then show that all group primitive languages satisfy the algebraic characterization of $\mathsf{SF}(\mathcal{B}q)$ which follows from~\cite{place2023closing}.
Consequently, all group primitive languages belong to $\mathsf{SF}(\mathcal{R})$.

For $L \subseteq \Sigma^*$ and $u \in \Sigma^*$, let $u^{-1}L \df \{ w \mid uw \in L \}$ and $L u^{-1} \df \{ w \mid wu \in L \}$.
A class $\mathcal{C}$ of languages is a \emph{prevariety} if $\mathcal{C}$ is a subset of the regular languages with $\emptyset, \Sigma^* \in \mathcal{C}$, for any $L \in \mathcal{C}$ and $u \in \Sigma^*$ we have $L u^{-1}, u^{-1} L \in \mathcal{C}$, and $\mathcal{C}$ is closed under union and complement.

\begin{exa}
\label{example:sf-prev}
Consider the \emph{star-free languages}, which we denote with $\mathsf{SF}$.
This class is the smallest class that contains the finite languages, and is closed under union, concatenation and complement.
Clearly, $\mathsf{SF}$ is contained in the regular languages, and it is straightforward to show that $\emptyset$ and $\Sigma^*$ belong to this class.
Furthermore, $\mathsf{SF}$ is closed under union and complement by definition.
Thus, $\mathsf{SF}$ is a prevariety since $L u^{-1}, u^{-1} L \in \mathsf{SF}$ for all $L \in \mathsf{SF}$ and $u \in \Sigma^*$; for example, see~\cite{pin2020prove}.
\end{exa}

\noindent
We say that a language $K \subseteq \Sigma^*$ \emph{separates} $L_1 \subseteq \Sigma^*$ from $L_2 \subseteq \Sigma^*$ if $L_1 \subseteq K$ and $K \intersect L_2 = \emptyset$.
For example, $\mathtt{a}^*$ separates $(\mathtt{aa})^*$ from $\mathtt{b}\cdot  \mathtt{a}^*$.
Then, for a class of languages $\mathcal{C}$, we say that $L_1$ is \emph{$\mathcal{C}$-separable} from $L_2$ if there exists $K \in \mathcal{C}$ that separates $L_1$ from $L_2$.
As an example, $(\mathtt{aa})^*$ and $\mathtt{a} \cdot (\mathtt{aa})^*$ are not $\mathsf{SF}$-separable as any $K \in \mathsf{SF}$ such that $(\mathtt{aa})^* \subseteq K$ necessarily has some word $w \in K \intersect \mathtt{a} \cdot (\mathtt{aa})^*$.
This follows from $(\mathtt{aa})^* \notin \mathsf{SF}$; see~\cite{pin2020prove}. 
The notion of separability is well-studied; see~\cite{almeida1999some, place2014separating} for starters.

\begin{defi}[Adapted from~\cite{place2023closing}]\label{defn:C-orbit}
Let $\mathcal{C}$ be a prevariety.
If $\mu \colon \Sigma^* \rightarrow M$ is a morphism where $M$ is a finite monoid, then for any idempotent element $i \in M$, the \emph{$\mathcal{C}$-orbit of $i$ with respect to $\mu$} is defined as the set of all elements $isi \in M$ where $s \in M$ such that $\mu^{-1}(s)$ is not $\mathcal{C}$-separable from $\mu^{-1}(i)$.
\end{defi}

\noindent
Note that $\mu^{-1}(s)$ is not $\mathcal{C}$-separable from $\mu^{-1}(i)$ if and only if $\mu^{-1}(s) \intersect L \neq \emptyset$ for any $L \in \mathcal{C}$ with $\mu^{-1}(i) \subseteq L$.

\begin{exa}
Consider the prevariety $\mathsf{SF}$ as discussed in~\autoref{example:sf-prev}.
Let $\mu \colon \{ \mathtt{a} \}^* \rightarrow \{ e, x \}$ such that $\mu(\mathtt{a}^n) = x$ if and only  if $n$ is odd.
Here, $e$ is the identity of $M$ and $x \cdot x = e$.
Clearly, $e$ is the only idempotent element of $M$.
The $\mathsf{SF}$-orbits of $e$ with respect to $\mu$ is the set of $ese \in M$ where $\mu^{-1}(s)$ is not $\mathsf{SF}$-separable from $\mu^{-1}(e)$.
Clearly, $e \cdot e \cdot e = e$ belongs to the $\mathsf{SF}$-orbits of $e$ with respect to $\mu$.
Furthermore, $e \cdot x \cdot e = x$ belongs to the  $\mathsf{SF}$-orbits of $e$ with respect to $\mu$ due to the fact that
$\mu^{-1}(e) = \{ \mathtt{a}^{2n} \mid n \in \mathbb{N} \}$ and $\mu^{-1}(x) = \{ \mathtt{a}^{2n+1} \mid n \in \mathbb{N} \}$ are not $\mathsf{SF}$-separable.
\end{exa}

\noindent
Lemma 5.5 of~\cite{place2023closing} states that for a prevariety $\mathcal{C}$, any $\mathcal{C}$-orbit is a monoid.
Furthermore, every morphism $\mu \colon \Sigma^* \rightarrow M$ gives rise to a set of monoids called the \emph{$\mathcal{C}$-orbits of $\mu$}; where every idempotent $i \in M$ is associated to some monoid in this set.
Using this notion of $\mathcal{C}$-orbits, Place and Zeitoun~\cite{place2023closing} algebraically characterized the languages belonging to $\mathsf{SF}(\mathcal{C})$, whenever $\mathcal{C}$ is a prevariety: 

\begin{thmC}[{\cite[Theorem~5.11]{place2023closing}}]\label{theorem:place-closing}
  Let $\mathcal{C}$ be a prevariety and let $L \subseteq \Sigma^*$ be regular. Then $L \in \mathsf{SF}(\mathcal{C})$ if and only if all the $\mathcal{C}$-orbits of the syntactic morphism are aperiodic monoids.
\end{thmC}

\noindent
In order to use~\autoref{theorem:place-closing}, we define the prevariety $\mathcal{B}q$ (introduced in~\cite{daviaud2018classes}): 
\begin{itemize}
\item $\emptyset, \Sigma^* \in \mathcal{B}q$,
\item $L \in \mathcal{B}q$ for every regular language $L \subseteq w^*$ with $w \in \Sigma^*$,
\item $\mathcal{B}q$ is closed under union and complement, and
\item $u^{-1}L, L u^{-1} \in \mathcal{B}q$ for any $L \in \mathcal{B}q$ and $u \in \Sigma^*$.
\end{itemize}

\noindent
Our next goal is to establish $\mathsf{SF}(\mathcal{R}) = \mathsf{SF}(\mathcal{B}q)$.
We start by giving two basic lemmas on bounded regular languages.
Recall that $L \subseteq \Sigma^*$ is a bounded regular language if $L$ is regular and $L \subseteq w_1^* \cdots w_n^*$ for $w_1, \dots, w_n \in \Sigma^+$ and $n \geq 1$.

\begin{lem}\label{lemma:bounded-regular}
Every bounded regular language $L \subseteq \Sigma^*$ belongs to  $\mathsf{SF}( \mathcal{R})$.
\end{lem}
\begin{proof}
From Theorem 1.1 of~\cite{ginsburg1966bounded}, the class of bounded regular languages is the smallest class that contains; the finite languages, the languages $w^*$ for any $w \in \Sigma^*$, and is closed under concatenation and finite union. 
Immediately from the definition of $\mathsf{SF}( \mathcal{R})$, we can see every bounded regular language $L \subseteq \Sigma^*$ is in $\mathsf{SF}(\mathcal{R})$.
\end{proof}

\begin{lem}\label{lemma:bounded-quotient}
If $L \subseteq w^*$, for some $w \in \Sigma^*$, is regular and $\mathtt a \in \Sigma$, then $\mathtt{a}^{-1} L$ and $L \mathtt{a}^{-1}$ are both bounded regular languages.
\end{lem}
\begin{proof}
Let $w \in \Sigma^*$ and let $I \subseteq \mathbb{N}$ such that $L = \{ w^i \mid i \in I \}$ is regular.
We shall prove that $\mathtt{a}^{-1} L$ is a bounded regular language ($L \mathtt{a}^{-1}$ follows symmetrically).
If $w = \mathtt{b} u$ for some $\mathtt b \in \Sigma \setminus \{ \mathtt a \}$ and $u \in \Sigma^*$, then $\mathtt{a}^{-1} L = \emptyset$ and therefore $\mathtt{a}^{-1} L$ is a bounded regular language.
If $w = \mathtt{a} u$ for some $u \in \Sigma^*$, then $\mathtt{a}^{-1} L = u \cdot \{ w^{i-1} \mid i \in I \text{ and } i > 0 \}$.
Thus, $\mathtt{a}^{-1} L \subset u^* w^*$ and thus is a bounded language.
Furthermore, if $L'$ is regular and $v \in \Sigma^*$, then $v^{-1} L'$ is also regular (see Theorem 4.5 of~\cite{hof97chp2}).
Consequently, $\mathtt{a}^{-1} L$ is bounded and regular.
\end{proof}

\noindent
Even though $\mathcal{R}$ is not a prevariety, the following lemma establishes that $\mathsf{SF}(\mathcal{R})$ can be characterized with $\mathcal{C}$-orbits.

\begin{lem}\label{lemma:SF-R-P}
$\mathsf{SF}(\mathcal{R}) = \mathsf{SF}(\mathcal{B}q)$.
\end{lem}
\begin{proof}
By the definition of $\mathcal{R}$ and $\mathcal{B}q$, we have that $\mathcal{R} \subseteq \mathcal{B}q$ and thus we immediately get $\mathsf{SF}(\mathcal{R}) \subseteq \mathsf{SF}(\mathcal{B}q)$.
To see that $\mathsf{SF}(\mathcal{B}q) \subseteq \mathsf{SF}(\mathcal{R})$, we show that every language $L \in \mathcal{B}q$ belongs to $\mathsf{SF}(\mathcal{R})$.
As stated earlier, $\emptyset, \Sigma^* \in \mathsf{SF}(\mathcal{R})$, see~\autoref{sec:star-free-closure-defn}.
The regular language $L \subseteq w^*$, for any $w \in \Sigma^*$, is a bounded regular language, thus $L \in \mathsf{SF}(\mathcal{R})$ (recall~\autoref{lemma:bounded-regular}).
Furthermore, $\mathsf{SF}(\mathcal{R})$ is closed under union and complement (by definition).
Therefore, all that remains is to prove:
\begin{clm}
If $L \in \mathcal{B}q$, then $u^{-1}L, L u^{-1} \in \mathsf{SF}(\mathcal{R})$ for any $u \in \Sigma^*$.
\end{clm}

\noindent
The proof of the above claim follows closely to the proof of Proposition 3.2 in~\cite{place2023closing}.
We only consider $u^{-1} L$, as $L u^{-1}$ follows symmetrically, and proceed by induction along the length of $u$.
If $|u| = 0$, then $u^{-1} L = L$ and therefore we are done.
Considering when $|u| > 0$, we have $u^{-1} L = w^{-1} ( \mathtt{a}^{-1} L)$ where $u = w \mathtt{a}$ and $\mathtt{a} \in \Sigma$.
Thus, by the induction hypothesis, $u^{-1} L \in \mathsf{SF}(\mathcal{R})$ if $\mathtt{a}^{-1} L \in \mathsf{SF}(\mathcal{R})$.
To show $\mathtt{a}^{-1} L \in \mathsf{SF}(\mathcal{R})$, we use a structural sub-induction.
Clearly $\mathtt{a}^{-1} \emptyset = \emptyset$ and $\mathtt{a}^{-1} \Sigma^* = \Sigma^*$ and therefore these cases are trivial.
The language $\mathtt{a}^{-1} L$ where $L \subseteq v^*$, where $L$ is regular and $v^* \in \Sigma^*$, is a bounded regular language (see~\autoref{lemma:bounded-quotient}) and therefore $\mathtt{a}^{-1}L \in \mathsf{SF}(\mathcal{R})$, see~\autoref{lemma:bounded-regular}. 
Next, consider $\mathtt{a}^{-1} L$ where $L = L_1 \union L_2$. 
By the sub-induction hypothesis, $\mathtt{a}^{-1} L_1 \in \mathsf{SF}(\mathcal{R})$ and $\mathtt{a}^{-1} L_2 \in \mathsf{SF}(\mathcal{R})$ and $\mathtt{a}^{-1}L = \mathtt{a}^{-1} L_1 \union \mathtt{a}^{-1} L_2$ and thus we are done with union.
For complement: Let $L = \Sigma^* \setminus K$ where $K \in \mathcal{B}q$. 
Note that
\[
\mathtt{a}^{-1} (\Sigma^* \setminus K) = \{ v \mid \mathtt{a}v \in \Sigma^* \setminus K \} = \\
\{ v \mid \mathtt{a} v \in \Sigma^* \text{ and } \mathtt{a} v \notin K \} = \Sigma^* \setminus \mathtt{a}^{-1} K. 
\]
By the sub-induction hypothesis $\mathtt a^{-1} K \in \mathsf{SF}(\mathcal{R})$, and thus we are done with complement.
This completes the proof that $\mathcal{B}q \subseteq \mathsf{SF}(\mathcal{R})$, which implies $\mathsf{SF}(\mathcal{B}q) \subseteq \mathsf{SF}(\mathcal{R})$.
\end{proof}

\noindent
We note that~\autoref{theorem:place-closing} and~\autoref{lemma:SF-R-P} immediately give us an algebraic characterization of $\mathsf{SF}(\mathcal{R})$.
However, for our more specialized purposes, we are able to give a more lightweight characterization.

\begin{lem}\label{lemma:group-prim-to-orbit}
If the syntactic morphism $\eta_L \colon \Sigma^* \rightarrow M_L$ of the regular language $L \subseteq \Sigma^*$ is a group primitive morphism, then all $\mathcal{B}q$-orbits of $\eta_L$ are aperiodic monoids.
\end{lem}
\begin{proof}
Let $\eta_L \colon \Sigma^* \rightarrow M_L$ be the syntactic morphism of the regular language $L \subseteq \Sigma^*$.
Working towards a contradiction, assume that $\eta_L$ is group primitive, and for some idempotent element $i \in M_L$, the $\mathcal{B}q$-orbit of $i$ with respect to $\eta_L$ contains some periodic element $isi \in M_L$; that is, $(isi)^n \neq (isi)^{n+1}$ for all $n \in \mathbb{N}_+$.
We fix such $i, s \in M_L$ for the remainder of this proof.
Since $i \in M_L$ is idempotent:
\[ (isi)^n = \underbrace{isi \cdot isi \cdots isi}_{n \text{-times}} = \underbrace{is \cdot is \cdots is}_{n \text{-times}} \cdot i = (is)^n i \]
for $n \in \mathbb{N}_+$.
Furthermore, since $(isi)^n \neq (isi)^{n+1}$ we know that $(is)^n i \neq (is)^{n+1} i$ for all $n \in \mathbb{N}_+$. 
Therefore $is \in M_L$ is a periodic element as $(is)^n \neq  (is)^{n+1}$ for all $n \in \mathbb{N}_+$.
We now consider four cases and show a contradiction for each.

\paragraph{Case 1, $\proot(\eta_L^{-1}(s)) = \emptyset$}
If this case holds, then $\eta_L^{-1}(s) = \{ \emptyword \}$ holds. 
This implies that $s$ is the identity element of $M_L$ and hence $(is)^n = i^n = i^{n+1} = (is)^{n+1}$, which contradicts our assumption that $is$ is periodic.

\paragraph{Case 2, $\proot(\eta_L^{-1}(i)) = \emptyset$}
Analogous to Case 1, it follows that $\eta_L^{-1}(i) = \{ \emptyword \}$.
Since $\{ \emptyword \} \in \mathcal{B}q$ by considering $\mathtt{a}^* \intersect \mathtt{b}^*$ with $\mathtt a \neq \mathtt b$, we know that $\eta_L^{-1}(s) \intersect \{ \emptyword \} \neq \emptyset$ by the definition of a $\mathcal{B}q$-orbit, see~\autoref{defn:C-orbit}.
Thus, $\emptyword \in \eta_L^{-1}(s)$ which implies that $s = \eta_L(\emptyword) = i$.
Hence, $i = s$ and consequently, $(is)^n = (ii)^n = i = (ii)^{n+1} = (is)^{n+1}$, which contradicts our assumption that $is$ is periodic.

\paragraph*{Case 3, $|\proot(\eta_L^{-1}(s)) \union \proot(\eta_L^{-1}(i))| = 1$}
Let $w \in \Sigma^+$ such that $\proot(\eta_L^{-1}(s)) = \{ w \} = \proot(\eta_L^{-1}(i))$.
Since $\eta_L^{-1}(i) \subseteq w^*$ is regular, we know that $\eta_L^{-1}(i) \in \mathcal{B}q$. 
Let $L' \df \eta_L^{-1}(i)$ and note $\eta_L^{-1}(i) \subseteq L'$.
Now, by the definition of a $\mathcal{B}q$-orbit, $\eta_L^{-1}(s) \intersect L' \neq \emptyset$.
Thus, $u \in \eta_L^{-1}(s) \intersect \eta_L^{-1}(i)$ for some $u \in \Sigma^*$.
This implies that $i = \eta_L(u) = s $.
Consequently, $i = s$, and hence $(is)^n = (ii)^n =(ii)^{n+1} = (is)^{n+1}$ as $i$ is idempotent.
This contradicts our assumption that $is \in M$ is periodic.

\paragraph*{Case 4, $|\proot(\eta_L^{-1}(s)) \union \proot(\eta_L^{-1}(i))| > 1$}
Let $u \in \eta_L^{-1}(i)$ and $v \in \eta_L^{-1}(s)$ such that $\proot(u) \neq \proot(v)$.
Note that $u$ and $v$ are both non-empty.
Since $i$ is idempotent, we have that $u u v, uv  \in \eta_L^{-1}(is)$, as $i i s = is$.
We know that $is \in M_L$ is periodic, and thus proving $\proot(uuv) \neq \proot(uv)$ would contradict our assumption that $\eta_L$ is group primitive.
Therefore, the rest of the proof of this case is dedicated to proving $\proot(uuv) \neq \proot(uv)$:
Working towards a contradiction, assume that $uv = r^i$ and $uuv = r^j$ for some $i,j \in \mathbb{N}_+$ and for some primitive word $r \in \Sigma^+$.
Notice that $(uv)^m = r^{2ij} = (uuv)^n$ where $m = 2j$ and $n = 2i$.
This implies that $v (uv)^{m-1} = uv (uuv)^{n-1}$.
Taking the prefix of length $|uv|$ from the left and right-hand side of the previous equality gives us $vu = uv$.
Proposition 1.3.2 from~\cite{lothaire1997combinatorics} states that $vu = uv$ implies $\proot(u) = \proot(v)$, which is a contradiction.

Observing that one of the above four cases must hold concludes this proof.
\end{proof}
\noindent
Immediately from~\autoref{theorem:place-closing}, \autoref{lemma:SF-R-P} and~\autoref{lemma:group-prim-to-orbit}, we get the main result of this section:

\begin{thm}\label{thm:gp-sf}
If $L \subseteq \Sigma^*$ is group primitive, then $L \in \mathsf{SF}(\mathcal{R})$.
\end{thm}
\begin{proof}
Let $\eta_L \colon \Sigma^* \rightarrow M_L$ be the group primitive syntactic morphism of the regular language $L \subseteq \Sigma^*$.
By~\autoref{lemma:group-prim-to-orbit}, we know that all $\mathcal{B}q$-orbits with respect to $\mu$ are aperiodic monoids.
Invoking~\autoref{theorem:place-closing}, we know that $L \in \mathsf{SF}(\mathcal{B}q)$, and by~\autoref{lemma:SF-R-P} this implies $L \in \mathsf{SF}(\mathcal{R})$.
\end{proof}

\noindent
Summarising the results from~\autoref{sec:sf-fc} and this section, we have that $\fc$ can define all the languages in $\mathsf{SF}(\mathcal{R})$, and that if a language is a group primitive language, then it can be expressed in $\mathsf{SF}(\mathcal{R})$.
An immediate corollary is that any group primitive language can be expressed in $\fc$.

\section{On Loop-Step Cycles and $\fc$-Definability}
\label{sec:fc-gp}
In this section, we show that any regular language
whose minimal DFA has a loop-step cycle (\autoref{defn:loop-step}) is \emph{not} expressible in $\fc$.

\subsection{Informal Proof Sketch}
A finite set $\{ a_1 < \cdots < a_n \} \subset \mathbb{N}$ is encoded as a word $\mathtt{1} \cdot \mathtt{0}^{a_1} \cdot \mathtt{1} \cdots \mathtt{1} \cdot \mathtt{0}^{a_n} \cdot \mathtt{1}$.
We shall use $\bar w_A \in \{ \mathtt{0}, \mathtt{1} \}^*$ to denote the word that encodes the set $A \subset \mathbb{N}$.
We then leverage a result from  Lynch~\cite{lynch1982sets} which provides sufficient  conditions for Duplicator to have a winning strategy for \emph{\ef games} over structures of the form $(\mathbb{N}, +, A)$ where $A \subseteq \mathbb{N}$, see~\autoref{thm:lynch}.
This can be converted to sufficient conditions for Duplicator to win \ef-games over two words that encode finite sets of natural numbers (\autoref{lemma:evenOnes}).
Once this is established, we show that for a morphism $h \colon \{ \mathtt{0}, \mathtt{1} \}^* \rightarrow \Sigma^*$ where $h(\mathtt{0})$ and $h(\mathtt{1})$ are two distinct primitive words of equal length, we have that $\bar w_A \equiv_{k+20} \bar w_B$ implies $h(\bar w_A) \equiv_k h(\bar w_B)$, assuming some lower bounds on the sizes of elements in $A$ and $B$; see~\autoref{lemma:morphGeneralize}.
Our last step is to show that if $\mathcal{M}$ is a minimal DFA and has a loop-step cycle, then for every $k \in \mathbb{N}$ there are sets $A,B \subset \mathbb{N}$ and words $p,s$ such that $p \cdot h(\bar w_A) \cdot s \in \mathcal{L}(\mathcal{M})$ and $p \cdot h(\bar w_B) \cdot s \notin \mathcal{L}(\mathcal{M})$, but $p \cdot h(\bar w_A) \cdot s \equiv_k p \cdot h(\bar w_B) \cdot s$, which concludes the proof~(\autoref{thm:fc-lsc}).

\subsection{Some Background on \ef Games}\label{subsec:EFgames}
The following definitions concerning \ef games follow closely to the definitions given in Chapter~3 of Libkin~\cite{libkin2004elements}.
\ef-games are played by two players called \emph{Spoiler} and \emph{Duplicator}.
The board consists of two relational structures $\Structure{A}$ and $\Structure{B}$ over the same signature.
They play for a predetermined number $k \in \mathbb{N}$ of rounds, and each round is as follows:
\begin{itemize}
\item Spoiler picks a structure $\Structure{A}$ or $\Structure{B}$, and then picks some element of their chosen structure.
\item Duplicator responds by picking some element in the structure that Spoiler did not choose.
\end{itemize}

\noindent
In order to define winning conditions, we first define a \emph{partial isomorphism}.

\begin{defi}[Partial Isomorphism]\label{defn:partialiso}
Let $\Structure{A}$ and $\Structure{B}$ be two relational structures over the same signature $\tau$, with the universes $A$ and $B$ respectively.
Let $\vec a = (a_1, \dots, a_n) \in A^n$ and $\vec b = (b_1, \dots, b_n) \in B^n$ .
Then $(\vec a, \vec b)$ defines a partial isomorphism between $\Structure{A}$ and $\Structure{B}$ if:
\begin{itemize}
\item $a_i = a_j \iff b_i = b_j$ for every $i,j \in [n]$;
\item $a_i = c^{\Structure{A}} \iff b_i = c^{\Structure{B}}$ for every constant symbol $c \in \tau$ and every $i \in [n]$.
\item $(a_{i_1}, \dots, a_{i_m}) \in R^\Structure{A} \iff (b_{i_1}, \dots, b_{i_m}) \in R^\Structure{B}$ for every relation symbol $R \in \tau$ with arity~$m$, and every sequence $(i_1, \dots, i_m) \in [n]^m$.
\end{itemize}
\end{defi}

\noindent
Let $\tau$ be a relational signature with constant symbols $c_1, \dots, c_r$.
Let $\Structure{A}$ and $\Structure{B}$ be two $\tau$-structures, and let $\mathcal{G}$ denote the $k$-round \ef game over $\Structure{A}$ and $\Structure{B}$.
Then, the \emph{resulting tuples} of $\mathcal{G}$ are $\vec a = (a_1, \dots, a_k, c_1^\Structure{A}, \dots, c_r^\Structure{A})$ and $\vec b = (b_1, \dots, b_k, c_1^\Structure{B}, \dots, c_r^\Structure{B})$ where $a_i \in A$ is the element of $\Structure{A}$ chosen in round $i$, and $b_i \in B$ is the element of $\Structure{B}$ chosen in round $i$ for each $i \in [k]$.
Then, Duplicator wins $\mathcal{G}$ if and only if $(\vec a, \vec b)$ forms a partial isomorphism.

For a $k$-round game $\mathcal{G}$ over $\Structure{A}$ and $\Structure{B}$, we say that Duplicator has a \emph{winning strategy} for $\mathcal{G}$ if and only if there is a way for Duplicator to play such that no matter what Spoiler chooses, Duplicator can win $\mathcal{G}$.
Otherwise, Spoiler has a winning strategy for $\mathcal{G}$.
If Duplicator has a winning strategy for $\mathcal{G}$, then we write $\Structure{A} \equiv_k \Structure{B}$. 
If $\Structure{A}$ and $\Structure{B}$ are $\fc$-structures, then $\Structure{A}$ and $\Structure{B}$ are uniquely determined (up to isomorphism) by some $w,v \in \Sigma^*$.
Therefore, we can simply write $w \equiv_k v$.

Each first-order formula $\varphi$ has a so-called \emph{quantifier rank}, denoted by the function $\qrank$.
Let $\varphi, \psi$ be first-order formulas over the same signature.
We define $\qrank$ recursively as follows:
\begin{itemize}
\item If $\varphi \in \fc$ is an atomic formula, then $\qrank(\varphi) \df 0$,
\item $\qrank(\neg \varphi) \df \qrank(\varphi)$, 
\item $\qrank(\varphi \land \psi) = \qrank(\varphi \lor \psi) \df \mathsf{max} \left( \qrank(\varphi), \qrank(\psi) \right)$, and 
\item $\qrank(Q x \colon \varphi) \df \qrank(\varphi) + 1$ for any $x \in \vars$ and $Q \in \{\forall, \exists\}$.
\end{itemize}
Let $\fc(k)$ denote the set of sentences $\varphi \in \fc$ such that $\qrank(\varphi) \leq k$.
Note if $w \not\equiv_k v$ for $w,v \in \Sigma^*$, then there exists $\varphi \in \fc(k)$ where $w \models \varphi$ and $v \not\models\varphi$; see Theorem 3.3 of~\cite{thompson2023generalized}.

It is known that winning strategies for \ef games completely characterize the expressive power of first-order logic.
We state this correspondence using a slight adaptation of Theorem~6.19 from Immerman~\cite{immerman1998descriptive}.

\begin{thm}\label{thm:ehrenfeucht}
  Let $\mathcal{C}$ be a class of finite or infinite structures over a finite signature, and let $P \subseteq \mathcal{C}$.
$P$ is not expressible in first-order logic if and only if for all $k \in \mathbb{N}$, there exist $\Structure{A}, \Structure{B} \in \mathcal{C}$ such that:
$\Structure{A} \in P$ and $\Structure{B} \notin P$, and
$\Structure{A} \equiv_k \Structure{B}$.
\end{thm}

\noindent
As $\fc$-structures are simply a class of structures over the signature $\signature_\Sigma$, it follows that~\autoref{thm:ehrenfeucht} immediately holds for $\fc$ (see~\cite{thompson2023generalized} for more details on \ef games over $\fc$-structures).

\subsection{A Regular Language Not Expressible in $\fc$}
\label{subsec:AddToConcat}
For our next step, we utilize \emph{Lynch's Theorem}~\cite{lynch1982sets}, which gives sufficient conditions for Duplicator to have a winning strategy over certain structures with addition.
We then connect this with $\fc$ by encoding these structures as words, which allows us to adapt the winning strategies given by Lynch's Theorem to inexpressibility results for $\fc$.

By $(\mathbb{N}, +,  A)$, we denote the structure with the universe $\mathbb{N}$, an addition relation $z = x + y$, and a finite set $A \subseteq \mathbb{N}$.\footnote{We note that Lynch~\cite{lynch1982sets} did not require $A$ to be finite.}

\begin{defiC}[\cite{lynch1982sets}]\label{defn:funcs}
Let $d$, $f$, and $g$ be functions from $\mathbb{N}$ to $\mathbb{N}$ defined as follows:
\begin{itemize}
\item $d(0) = 5$, and $d(i+1) = (2^{i+3} + 1) d(i)$,
\item $f(0) = 1$, and $f(i+1) = 2f(i)^4$, and
\item $g(0) = 0$, and $g(i+1) = 2 f(i)^2 g(i) + f(i)!$.
\end{itemize}
\end{defiC}

\noindent
The exact definitions of $d$, $f$, and $g$ are not important for our results.
We only give their definitions for some intuition regarding the following:

\begin{defi}[Adapted from Lynch~\cite{lynch1982sets}]\label{defn:Lynchian}
Let $k \in \mathbb{N}$ and let $(p_i)_{i \in \mathbb{N}}$ be any sequence in $\mathbb{N}$ such that $p_0 = 0$ and 
\begin{equation}\label{eq:1}
p_{i+1} \geq 2^{k+3} f(k)^3 p_i + 2 f(k)^2 g(k) 
\end{equation}
for $i \in \mathbb{N}$, and $p_i \equiv p_j \pmod{f(k)!}$ for $i, j \in \mathbb{N}_+$.
Then, any set of the form $P_k = \{ p_i \mid i \geq 1 \}$ is a \emph{$k$-Lynchian set}.
\end{defi}

\noindent
We are now ready to state Lynch's Theorem:
\begin{thm}[Lynch's Theorem~\cite{lynch1982sets}]\label{thm:lynch}
Let $k \in \mathbb{N}$ and let $P_k \subseteq \mathbb{N}$ be a $k$-Lynchian set.
For any finite sets $A, B \subseteq P_k$ where $d(k) < |A|, |B| $ we have that $(\mathbb{N}, +, A) \equiv_k (\mathbb{N}, +, B)$. 
\end{thm}
\noindent
Informally, Lynch's Theorem states that as long as $A$ and $B$ are big enough, and only contain elements from a $k$-Lynchian set, we have that $(\mathbb N, +, A) \equiv_k (\mathbb N, +, B)$.
As an example, for any $k \in \mathbb{N}$ we can pick $A$ and $B$ such that $(\mathbb{N}, +, A) \equiv_k (\mathbb N, +, B)$ where $|A|$ is even, $|B|$ is odd.

Our next focus is on adapting \ef game strategies over structures with addition to $\fc$-structures.
This will provide us with new insights into limitations of the expressive power of $\fc$.
Namely, we shall prove that there exist some regular languages that cannot be expressed by~$\fc$. 

\begin{defi}\label{defn:setWord}
For every finite set $A \subset \mathbb{N}$, the canonical encoding of $A = \{ a_1 < \dots < a_n \}$ as a word $\bar w_A \in \{ \mathtt 0, \mathtt 1 \}^*$ is $\bar w_A \df \mathtt{1} \cdot \mathtt{0}^{a_1} \cdot \mathtt{1} \cdot \mathtt{0}^{a_2} \cdot \mathtt{1} \cdots \mathtt{1} \cdot \mathtt{0}^{a_n} \cdot \mathtt{1}$.
\end{defi}

\noindent
Let $\mathsf{FO}[+,S]$ denote first-order logic over the class of structures of the form $(\mathbb{N}, + , S)$ where $S \subseteq \mathbb{N}$.
Using the canonical encoding of a set as a word, we can translate formulas over $\fc$-structures to $\mathsf{FO}[+,S]$.

\begin{lem}\label{lemma:evenOnes}
  There exists a function $r\colon\NN\to\NN$ such that for every $k\in\NN$ and all finite sets
  $A,B\subset \NN$ the following holds: If  $(\mathbb{N}, + , A) \equiv_{r(k)} (\mathbb{N}, +, B)$ then $\bar w_A \equiv_{k} \bar w_B$.  
\end{lem}
\begin{proof}
  Consider an arbitrary $k\in\NN$ and arbitrary finite sets $A,B\subset\NN$.
Let $\Structure{A} \df (\mathbb{N}, +, A)$ and $\Structure{B} \df (\mathbb{N}, +, B)$.

To prove the result, we prove the contraposition, i.e., we show that if 
$\bar w_A \not\equiv_k \bar w_B$, then $\Structure{A} \not\equiv_{r(k)} \Structure{B}$, for a suitably chosen number $r(k)$ that does not depend on the sets $A,B$.

Assume that $\bar w_A \not\equiv_k \bar w_B$, that is, there exists an $\fc$-sentence
$\varphi \in \fc(k)$ such that $\bar w_A \models \varphi$ and $\bar w_B \not\models \varphi$.
In the following, we construct a first-order reduction that translates
a structure $(\NN,+,A)$ into the word $\bar w_A$, and which can thus be used to translate
$\varphi$ into a
sentence $\psi \in \mathsf{FO}[+,S]$ such that $\Structure{A} \models \psi$ and $\Structure{B} \not\models \psi$.
That is, using formulas in $\mathsf{FO}[+,S]$, we simulate all the behaviour of an $\mathsf{FC}$-formula over words that canonically encode a set (as in~\autoref{defn:setWord}). 
See Chapter 3 of Immerman~\cite{immerman1998descriptive} for more details on first-order reductions.
Finally, we will be done by letting $r(k)\deff \qr(\psi)$, since the quantifier rank of $\psi$ only depends on the first-order reduction and the quantifier rank $k$ of the formula $\varphi$.

We extend $\{+,S\}$-structures with a new constant $\bot$ which does not appear in the addition relation; clearly this does not change the expressive power of $\mathsf{FO}[+,S]$.
Let $\bar w_A \in \{ \mathtt 0, \mathtt 1 \}^*$ be the canonical encoding of a non-empty and finite set $A = \{ x_1 < x_2 < \dots < x_{|A|} \} \subset \mathbb{N}$.
For every factor $u \sqsubseteq \bar w_A$, we associate a tuple $t_u =  (a_1, a_2, a_3, a_4) \in (\mathbb{N} \union \{ \bot \})^4$:
\begin{itemize}
\item If $|u|_\mathtt{1} \geq 2$, then $u = \mathtt{0}^m \cdot \mathtt{1} \cdot \mathtt{0}^{x_l} \cdot \mathtt{1} \cdots \mathtt{1} \cdot \mathtt{0}^{x_r} \cdot \mathtt{1} \cdot \mathtt{0}^n$, where $1 \leq l \leq r \leq |A|$. Thus, we let $t_u \df (m, x_l, x_r, n)$.
\item If $|u|_\mathtt{1} = 1$, then $u = \mathtt{0}^m \mathtt{1} \mathtt{0}^n$. Thus, $t_u \df (m, \bot, \bot, n)$.
\item If $|u|_\mathtt{0} = 0$, then $u = \mathtt{0}^m$. Thus, $t_u \df (m, \bot, \bot, \bot)$.
\end{itemize}

\paragraph*{Universe}
Let $\psi_{\mathsf{univ}}(x_1, x_2, x_3, x_4)$ be a formula in $\mathsf{FO}[+,S]$ such that for any $\Structure{A}$, we have
\[
\psi_\mathsf{univ}(\Structure{A}) = \{ (a_1,a_2,a_3,a_4) \in  (\mathbb{N} \union \{ \bot \})^4 \mid (a_1, a_2, a_3, a_4) = t_u \text{ for some } u \sqsubseteq \bar w_A \}. 
\]
It is easy to see that $\psi_\mathsf{univ}(x_1,x_2,x_3,x_4)$ can be written in $\mathsf{FO}[+,S]$ using the following description (and sketch for the case where $|u|_\mathtt{a} =1$) of the formula:

To deal with the case where $|u|_\mathtt{1} = 1$, consider
\begin{multline*}
\mathsf{univ}_1(x_1, x_2, x_3, x_4) \df (x_2 \logeq \perp) \land (x_3 \logeq \perp) \land (x_4 \leq \mathsf{max}(A)) \land \\
\exists y \colon \Bigl( A(y) \land \neg \exists z \colon \bigl( A(z) \land (y < z < \mathsf{max}(A) ) \land (x_1 \leq y) \bigr) \Bigr).
\end{multline*}
The meaning of the top line of $\mathsf{univ}_1(x_1, x_2, x_3, x_4)$ is obvious.
The second line states that $x_1$ is smaller than (or equal to) the second largest value of $A$.

To deal with the case where $|u|_\mathtt{1} \geq 2$, the formula $\psi_\mathsf{univ}(x_1, x_2, x_3, x_4)$ states the following:
\begin{itemize}
\item $x_2, x_3 \in A$ and $x_2 < x_3$,
\item if $x_2 = \mathsf{min}(A)$ then $x_1 = 0$, otherwise $x_1$ is smaller than or equal to the largest element of $A$ less than $x_2$, and
\item if $x_3 = \mathsf{max}(A)$ then $x_4 = 0$, otherwise $x_4$ is smaller than or equal to the smallest element of $A$ larger than $x_3$.
\end{itemize}

\noindent
Lastly, to deal with the case where $|u|_\mathtt{a} = 0$,  the formula $\psi_\mathsf{univ}(x_1, x_2, x_3, x_4)$ states that $x_2 = x_3 = x_4 = \perp$ and that $x_1 \leq \mathsf{max}(A)$.

\paragraph*{Constants}
Since the constant symbol $\mathtt 0$ is represented by the tuple $(1, \bot, \bot, \bot)$, let
\[ \psi_\mathtt{0}(x_1, x_2, x_3, x_4) \df (x_1 \logeq 1) \land \bigwedge_{2 \leq i \leq 4} (x_i \logeq \bot). \]

\noindent
Since the constant symbol $\mathtt 1$ is represented by $(0, \bot, \bot, 0)$, let 
\[\psi_\mathtt{1}(x_1, x_2, x_3, x_4) \df (x_1 \logeq 0) \land (x_4 \logeq 0) \land (x_2 \logeq \bot) \land (x_3 \logeq \bot).\]

\paragraph*{Concatenation}
To simulate concatenation, we define the ternary relation $\oplus$ over elements of $\psi_\mathsf{univ}(\Structure{A})$.
We wish to have that $t_{w} = t_u \oplus t_v$ if and only if $w = uv$.
To realize this behaviour, we look at various cases. 
However, we shall only look at one case in detail, as the others follow analogously.
\begin{itemize}
\item Case 1, $|u|_\mathtt{1}, |v|_\mathtt{1} \geq 2$.
For this case, we have that $u = \mathtt{0}^{m} \mathtt{1} \mathtt{0}^{x_l} \mathtt{1} \cdots \mathtt{1} \mathtt{0}^{x_r} \mathtt{1} \mathtt{0}^n$ and $v = \mathtt{0}^{m'} \mathtt{1} \mathtt{0}^{x_{l'}} \mathtt{1} \cdots \mathtt{1} \mathtt{0}^{x_{r'}} \mathtt{1} \mathtt{0}^{n'}$. 
Then, $u \cdot v = \mathtt{0}^m \mathtt{1} \mathtt{0}^{x_l} \mathtt{1} \cdots \mathtt{0}^{x_{r'}} \mathtt{1} \mathtt{0}^{n'}$ is a factor of $\bar w_A$, if and only if $n+m' = x_{r+1} = x_{l'-1}$.
Thus 
\[ \underbrace{(m, x_l, x_{r'}, n')}_{t_{uv}} = \underbrace{(m,x_l,x_r,n)}_{t_u} \oplus \underbrace{(m', x_{l'}, x_{r'}, n')}_{t_v}  \iff n+m' = x_{r+1} = x_{l'-1} \]
It is clear that one could write an $\mathsf{FO}[+, S]$-formula to realize this behaviour:
Consider a formula with arity $12$, which first ensures the first four, the middle four, and the last four components satisfy $\psi_{\mathsf{univ}}(\Structure A)$.
The formula also ensures that none of the components are $\bot$, and the required arithmetic and 
equalities holds.

\item Case 2, $|u|_\mathtt{1} \geq 2$  and $|v|_\mathtt{1} = 1$. 
\[ (m,x_l,x_{r+1},n') = (m, x_l, x_r, n) \oplus (m', \bot, \bot, n') \iff n+m' = x_{r+1}. \]

\item Case 3, $|u|_\mathtt{1} = 1$  and $|v|_\mathtt{1} \geq 2$. Symmetric to Case 2.

\item Case 4, $|u|_\mathtt{1} \geq 2$ and $|v|_\mathtt{1} = 0$. 
\[ (m,x_l,x_r,n+m') = (m,x_l,x_r,n) \oplus (m',\bot, \bot, \bot) \iff n+m' \leq x_{r+1}. \]

\item Case 5, $|u|_\mathtt{1} = 0$ and $|v|_\mathtt{1} \geq 2$. Symmetric to Case 4.

\item Case 6, $|u|_\mathtt{1} = 1$ and $|v|_\mathtt{1} = 1$.  
\begin{align*}
& (m,n+m', n+m', n')  = (m,\bot,\bot,n) \oplus (m',\bot,\bot,n') \\
\iff & n+m' = x_i \in A \text{ where } m \leq x_{i-1} \text{ and } n' \leq x_{i+1}.
\end{align*}

\item Case 7, $|u|_\mathtt{1} = 1$ and $|v|_\mathtt{1} = 0$. 
\[ (m, \bot, \bot, n+m') =   (m, \bot, \bot, n) \oplus (m', \bot,\bot,\bot) \iff n+m' \leq \mathsf{max}(A). \]

\item Case 8, $|u|_\mathtt{1} = 0$ and $|v|_\mathtt{1} = 1$. Symmetric to Case 7.

\item Case 9, $|u|_\mathtt{1} = 0$ and $|v|_\mathtt{1} = 0$. 
\[ (m+m', \bot, \bot, \bot) = (m, \bot, \bot, \bot) \oplus (m', \bot, \bot, \bot) \iff m+m' \leq \mathsf{max}(A). \]
\end{itemize}
It is a straightforward exercise to define a $\mathsf{FO}[+,S]$-formula for each of the above cases (the formula can distinguish between the different cases by seeing which components are $\perp$, then the formula only needs to check that the required equalities, inequalities, and memberships in $A$ hold).
Combining the formulas for each case with disjunction results in a formula that has the correct behaviour for $\oplus$.
That is, it is easily observed that the above definition of $\oplus$ gives us $t_{w} = t_u \oplus t_v \iff w = uv$ for any $w,u,v \sqsubseteq \bar w_A$.

One can now rewrite $\varphi \in \fc(k)$ that separates $\bar w_A$ and $\bar w_B$ in a straightforward manner, using the above defined concepts, into a formula $\psi \in \mathsf{FO}[+,S]$ such that $\psi$ separates $\Structure{A}$ and $\Structure{B}$. In the literature, this is achieved by applying the ``interpretation lemma'' associated with the notion of first-order reductions. To keep the paper self-contained in this respect, instead of formulating and applying the ``interpretation lemma'', in the following we explicitly describe how $\psi$ is obtained from $\varphi$.
We define a mapping $\mathcal{T} \colon \fc \rightarrow \mathsf{FO}[+,S] $.

First, let $\ell \colon \vars \union \{ \mathtt{0}, \mathtt{1} \} \rightarrow (\Gamma \union \{ 0, 1, \bot \} )^4$ be a function that maps each $\fc$ variable $x \in \vars$ to a unique tuple of variables $\vec x \in \Gamma^4$ (where $\Gamma$ is a countably infinite set of variables that is disjoint from $\vars$), maps $\mathtt{0}$ to $(1,\bot,\bot,\bot)$ and maps $\mathtt{1}$ to $(0, \bot, \bot, 0)$.
Note that as $\ell(\mathtt 0)$ and $\ell(\mathtt 1)$ can easily be expressed with an $\mathsf{FO}[+,S]$-formula, we can treat them as new constant symbols.
Likewise, we treat $\oplus$ as a relational symbol as it can also be expressed in $\mathsf{FO}[+,S]$.
For $\varphi, \varphi' \in \fc$, let 
\begin{itemize}
\item $\mathcal{T}(Q x \colon \varphi) \df Q \vec x  \colon ( \psi_\mathsf{univ}(\vec x ) \land \mathcal{T}(\varphi) )$ for $Q \in \{\forall, \exists \}$ and $x \in \Xi$,
\item $\mathcal{T}(\varphi \land \varphi') \df \mathcal{T}(\varphi) \land \mathcal{T}(\varphi')$,
\item $\mathcal{T}(\varphi \lor \varphi') \df \mathcal{T}(\varphi) \lor \mathcal{T}(\varphi')$,
\item $\mathcal{T}(\neg \varphi) \df \neg \mathcal{T}(\varphi)$,
\item $\mathcal{T} (x \logeq y \cdot z) \df \bigl( \ell(x) \logeq \ell(y) \oplus \ell(z) \bigr)$.
\end{itemize}
Here, if $\vec x = (y_1, y_2, y_3, y_4)$ and $Q \in \{ \exists, \forall \}$, we use $Q \vec x$ as shorthand for $Q y_1, y_2, y_3, y_4$.

Then, for any $\varphi \in \fc$ such that $\bar w_A \models \varphi$ and $\bar w_B \not\models \varphi$, we have that $\Structure{A} \models \mathcal{T}(\varphi)$ and $\Structure{B} \not\models \mathcal{T}(\varphi)$.
Consequently, there exists some $\mathsf{FO}[+,S]$-formula $\psi$ that separates $\Structure{A}$ from $\Structure{B}$, which means that $\Structure{A} \not\equiv_{r(k)} \Structure{B}$ where $r(k) = \mathsf{qr}(\psi)$. 
This completes the proof of \autoref{lemma:evenOnes}.
\end{proof}

\noindent
By combining \autoref{thm:lynch} and \autoref{lemma:evenOnes}, we obtain that there are regular languages that cannot be expressed in $\fc$.

\begin{thm}
$L\deff\setc{w\in \set{\mathtt{0},\mathtt{1}}^*}{|w|_{\mathtt{1}} \text{ is even}}$ is a regular language that is not $\fc$-definable. 
\end{thm}
\begin{proof}
  Clearly, $L$ is a regular language. In order to show that $L$ is not $\fc$-definable,
  by \autoref{thm:ehrenfeucht} it suffices to construct, for every $k\in\NN$, two words $u_k\in L$ and $v_k\not\in L$ such that $u_k\equiv_k v_k$.

  Let us fix an arbitrary $k\in\NN$, and let $k'\deff r(k)$ be the number provided by
  \autoref{lemma:evenOnes}.
  Let $P_{k'} \subseteq \mathbb{N}$ be a $k'$-Lynchian set.
 From \autoref{thm:lynch} we know that
for finite sets $A, B \subseteq P_{k'}$ where $d(k') < |A|, |B|$ we have $(\mathbb{N}, +, A) \equiv_{k'} (\mathbb{N}, +, B)$.
Thus, the only restriction on $|A|$ and $|B|$ is that they are large enough.
Here, we choose finite sets $A,B\subseteq P_{k'}$ where where $|A|$ is even and $|B|$ is odd and where $d(k') < |A|, |B|$.
It follows from \autoref{thm:lynch} that $(\mathbb{N}, +, A) \equiv_{k'} (\mathbb{N}, +, B)$.
As $k'=r(k)$, by \autoref{lemma:evenOnes} we obtain: $\bar w_A \equiv_{k} \bar w_B$, where $\bar w_A, \bar w_B \in \{\mathtt 0, \mathtt 1 \}^*$ are the canonical encodings of $A$ and $B$ respectively.

Notice since $|A|$ is even and $|B|$ is odd, 
$|\bar w_A|_{\mathtt 1}$ is odd and $|\bar w_B|_{\mathtt 1}$ is even.
Thus, $\bar w_B \in L$, $\bar w_A \not\in L$ and $\bar w_B\equiv_k \bar w_A$. 
\end{proof}

\noindent
The technique used in the above proof (for different choices of the sets $A,B$)
actually provides us with a range of further regular languages that are not $\fc$-definable.
In particular, 
the connection between addition and concatenation provided by
\autoref{lemma:evenOnes} 
is a key step that we will also use in the following subsections in order to prove that the language of \emph{any} minimal DFA that has a loop-step cycle is not $\fc$-definable.

\subsection{Generalizing With Morphisms}
Our next step is to generalize~\autoref{lemma:evenOnes} using morphisms of the form $h \colon \{ \mathtt 0, \mathtt 1 \}^* \rightarrow \Sigma^*$.
Informally, our goal is to show that $\bar w_A \equiv_k \bar w_B$ implies $h(\bar w_A) \equiv_{k'} h(\bar w_B)$, for a specific type of morphism $h$.

We call a morphism $h \colon \{ \mathtt 0, \mathtt 1 \}^* \rightarrow \Sigma^*$ a \emph{uniform morphism} if $h(\mathtt 0)$ and $h(\mathtt 1)$ are two distinct primitive words of equal length.

Before moving on, let us state a folklore result, see Section 2.2, Chapter 6 of~\cite{handbookOfFL} for more details of~\autoref{lemma:internal}:
\begin{lem}[Folklore]\label{lemma:internal}
A word $w \in \Sigma^+$ is an internal factor of $ww$ if and only if $w$ is imprimitive.
\end{lem}

\noindent
Next, we introduce the following concept:

\begin{defi}
Let $\Sigma$ be any alphabet where $|\Sigma| \geq 2$.
Let $h \colon \{ \mathtt 0, \mathtt 1 \}^* \rightarrow \Sigma^*$ be a uniform morphism and let $w \in \{\mathtt 0, \mathtt 1 \}^*$.
Then, for any $u \sqsubseteq h(w)$, we say that $(x, y, z) \in \Sigma^* \times \{\mathtt 0, \mathtt 1 \}^* \times \Sigma^*$
is a \emph{core factorization of $u$ with respect to $w$ and $h$} if $u = x \cdot h(y) \cdot z$ and:
\begin{itemize}
\item $x \ssuff h(c)$ for some $c \in \{\mathtt 0, \mathtt 1\}$, 
\item $y \sqsubseteq w$, and
\item $z \spref h(c')$ for some $c' \in \{ \mathtt 0, \mathtt 1 \}$.
\end{itemize}
If $(x, y, z)$ is a core factorization of $u$ with respect to $w$ and $h$, then $y \sqsubseteq w$ is a \emph{pre-image core of $u$} with respect to $w$ and~$h$.
\end{defi}
\noindent
When $w \in \Sigma^*$ and $h \colon \{ \mathtt 0, \mathtt 1 \}^* \rightarrow \Sigma^*$ are clear from context, we simply say $(x,y,z)$ is a core factorization of $u$, and $y$ is a pre-image core of $u$.

\begin{exa}
Let $h \colon \{ \mathtt{0}, \mathtt{1} \}^* \rightarrow \{ \mathtt{a}, \mathtt{b} \}^*$ with $h(\mathtt{0}) = \mathtt{aba}$ and $h(\mathtt{1}) = \mathtt{bb}$.
Then, $h(\mathtt{0100}) = \mathtt{aba} \cdot \mathtt{bb} \cdot \mathtt{aba} \cdot \mathtt{aba}$.
It follows that $(\mathtt{a}, \mathtt{1}, \mathtt{a})$ is a core factorization of $\mathtt{a bb a}$.
\end{exa}

\noindent
Recall~\autoref{defn:setWord} of a canonical encoding of a finite set $A \subset \mathbb{N}$ as a word $\bar w_A \in \{ \mathtt 0, \mathtt 1 \}^*$.
As a quick example, the set $A = \{1,2,4 \}$ is encoded as $ \bar w_A = \mathtt 1 \cdot \mathtt{0} \cdot \mathtt {1} \cdot \mathtt{0}^2 \cdot \mathtt{1} \cdot \mathtt{0}^4 \cdot \mathtt{1}$.

\begin{lem}\label{lemma:unique-core}
Let $\Sigma$ be any alphabet where $|\Sigma| \geq 2$.
Let $h \colon \{ \mathtt 0, \mathtt 1 \}^* \rightarrow \Sigma^*$ be a uniform morphism.
Let $\bar w_A \in \{\mathtt 0, \mathtt 1 \}^*$ be the canonical encoding of a finite set $A \subset \mathbb{N}$ where $\mathsf{min}(A) \geq 2 |h(\mathtt 1)|$.
Then, for any $u \sqsubseteq h(\bar w_A)$ where $|u|> 10 \cdot |h(\mathtt 1)|$, the pre-image core of $u$ is unique.
\end{lem}
\begin{proof}
Let $\mathsf{m} \df |h(\mathtt 0)|=|h(\mathtt 1)|$.
Since $h(\mathtt 0) \neq h(\mathtt 1)$ but $|h(\mathtt 0)|= |h(\mathtt 1)|$, the morphism $h$ is injective.
We shall show that assuming $x \cdot h(y) \cdot z=x' \cdot h(y') \cdot z'$ are two distinct core factorizations of $u$ results in a contradiction. 
Since $x,x'$ are proper suffixes of elements of $\{ h(\mathtt 0), h(\mathtt 1) \}$, and $z,z'$ are proper prefixes of elements of $\{h(\mathtt 0), h(\mathtt 1) \}$, we have that
$0\leq |x|,|x'|,|z|,|z'|<\mathsf{m}$.

\paragraph{Case 1, $x = x'$:}
It follows that $h(y) \cdot z = h(y') \cdot z'$ and thus $|h(y)| - |h(y')| = |z'| - |z|$. 
The right-hand side has an absolute value which is strictly smaller than $\mathsf{m}$, while the left-hand side is a multiple of $\mathsf{m}$. 
This only holds when both are zero and thus $|z| = |z'|$.
Therefore, $x=x'$, $z=z'$, and $h(y)=h(y')$.
Since $h$ is injective, we obtain $y=y'$.

\paragraph{Case 2, $x \neq x'$:} 
Suppose, without loss of generality, that $x' = x \cdot x_2$ for some $x_2 \in \Sigma^+$.
Since $x$ and $x'$ are proper suffixes of $h(\mathtt 0)$ or $h(\mathtt 1)$, we have that 
$1 \leq |x_2| < \mathsf{m}$.
Thus, $h(y) \cdot z = x_2 \cdot h(y') \cdot z'$.
Since $|u| > 10\mathsf{m}$, we have that $|y| \geq 9$.
Moreover, $y$ contains an occurrence of $\mathtt{0}^3$ since $\mathsf{min}(A) \geq 2\mathsf{m}$.
Therefore, $y = y_1 \cdot \mathtt{0}^3 \cdot y_2$ for some $y_1, y_2 \in \{ \mathtt{0}, \mathtt{1} \}^*$.

\begin{figure}
    \centering
    \begin{tikzpicture}
        \draw[draw=black] (-4, 0) rectangle ++(8,1);
        \draw[draw=black] (-4, -1) rectangle ++(8,1);

        \draw[draw=black] (-0.5, 0) -- (-0.5, 1);
        \draw[draw=black] (0.5, 0) -- (0.5, 1);
        \draw[draw=black] (-1.5,0) -- (-1.5, 1);
        \draw[draw=black] (1.5,0) -- (1.5, 1);
        
        \node(1) at (-1, 0.5) {$h(\mathtt 0)$};
        \node(2) at (1, 0.5) {$h(\mathtt 0)$};
        \node(3) at (0, 0.5) {$h(\mathtt 0)$};

        \node(5) at (-2.5, 0.5) {$h(y_1)$};

        \node(6) at (2.5, 0.5) {$\cdots$};        

        \draw[draw=black] (0, -1) -- (0,0);
        \draw[draw=black] (-1, -1) -- (-1,0);
        \draw[draw=black] (1, -1) -- (1,0);

        \draw[draw=black] (-3,-1) -- (-3,0);

        \node(4) at (-0.5, -0.5) {$h(s)$};
        \node(5) at (0.5, -0.5) {$h(t)$};
        \node(6) at (-3.5, -0.5) {$x_2$};
        \node(7) at (-2, -0.5) {$h(y_1')$};
        \node(6) at (2.5, -0.5) {$\cdots$};

    \end{tikzpicture}   
    \caption{Illustration of the second case of~\autoref{lemma:unique-core}}
    \label{fig:proof-illustration}
\end{figure}

Comparing lengths in $h(y) \cdot z = x_2 \cdot h(y') \cdot z'$, we get 
$|h(y)| - |h(y')| = |x_2| + |z'| - |z|$.
Notice that $|h(y)| - |h(y')| = \mathsf{m} |y| - \mathsf{m} |y'|$ and thus $\mathsf{m} |y| - \mathsf{m} |y'| = |x_2| + |z'| - |z|$.
Since $0 < |x_2| < \mathsf{m}$ and $0 \leq |z|, |z'| < \mathsf{m}$, it follows that
$\mathsf{m} |y| - \mathsf{m} |y'| \leq  2\mathsf{m}-2$ and thus $|y'| \geq |y| -1 $.

Moreover, since $|y| \geq |y_1| + 3$, we get
$|y'| \geq |y_1| + 2$.
Thus, we may factorize $y'$ as $y_1' \cdot s \cdot t \cdot y_2'$ where $s, t \in \{ \mathtt{0}, \mathtt{1} \}$ and $|y_1'| = |y_1|$. 
Hence,
\[
  \underbrace{h(y_1) \cdot h(\mathtt 0)^3 \cdot h(y_2)}_{= h(y)} \cdot z
  \ = \
  x_2 \cdot \underbrace{h(y_1') \cdot h(s) \cdot h(t) \cdot
    h(y_2')}_{= h(y')} \cdot z' .
\]
Note that on the right hand side, the factor $h(s)\cdot h(t)$ starts
directly right to the $|x_2 \cdot h(y_1')|$-th letter of $x_2 \cdot
h(y') \cdot z'$, and on the left hand side, the factor $h(\mathtt 0)^3$ starts
directly right of the $|h(y_1)|$-th letter of $h(y)$.
Since $|x_2\cdot h(y_1')|=|x_2|+\mathsf{m}|y_1'|=|x_2|+\mathsf{m}|y_1|=|x_2|+|h(y_1)|$
and $1\leq|x_2|<\mathsf{m}$, we obtain that $h(s)\cdot h(t)$ is an internal
factor of $h(\mathtt 0)^3$;  see~\autoref{fig:proof-illustration} for an illustration.
Moreover, $\bar w_A$ does not contain $\mathtt{11}$ as a factor,
and hence we have
$s = \mathtt 0$ or $t = \mathtt 0$.
This implies that $h(\mathtt 0)$ is an internal factor of $h(\mathtt
0) \cdot h(\mathtt 0)$ which contradicts the fact that $h(\mathtt 0)$
is primitive (cf.~\autoref{lemma:internal}).
\end{proof}

\noindent
For our next step, we use a basic lemma from~\cite{thompson2023generalized}.
Informally, this result states that if Spoiler picks a short factor (with respect to the number of remaining rounds), then Duplicator must respond with the identical factor (or lose).

\begin{lemC}[\cite{thompson2023generalized}]\label{lemma:consistentStrats}
Let $\Structure{A}_w$ and $\Structure{B}_v$ be $\signature_\Sigma$-structures that represent $w \in \Sigma^*$ and $v \in \Sigma^*$, where $\Structure{A}_w \equiv_k \Structure{B}_v$.
Let $\vec a = (a_1, a_2, \dots, a_{k+|\Sigma|+1})$ and $\vec b= (b_1, b_2, \dots, b_{k+|\Sigma|+1})$ be the tuple resulting from a $k$-round game over $\Structure{A}_w$ and $\Structure{B}_v$ where Duplicator plays their winning strategy.
If $|a_r|< k - r + 1$ or $ |b_r| < k - r + 1$ for some $r \in [k]$, then $b_r = a_r$. 
\end{lemC}

\noindent
Let $\Sigma$ be any alphabet where $| \Sigma | \geq 2$ and let $h \colon \{ \mathtt 0, \mathtt 1 \}^* \rightarrow \Sigma^*$ be a uniform morphism.
Let $\bar w_A \in \{ \mathtt 0, \mathtt 1 \}^*$ be the canonical encoding of a set $A \subset \mathbb{N}$ where the $\mathsf{min}(A) > 2 \cdot |h(\mathtt 1)|$.
Then, from~\autoref{lemma:unique-core}, we know that the core-factorization of $u \sqsubseteq h(\bar w_A)$ is unique.
This gives us a way to associate every long enough factor of $h(\bar w_A)$ with some factor of $\bar w_A$, and hence gives a tool for generalizing Duplicator's strategy over words of the form $\bar w_A$.

\begin{lem}\label{lemma:morphGeneralize}
Let $h \colon \{ \mathtt 0, \mathtt 1 \}^* \rightarrow \Sigma^*$ be a uniform morphism, where $| \Sigma | \geq 2$.
If $\bar w_A \equiv_{k+20} \bar w_B$ for sets $A, B \subseteq \mathbb{N}$ where $\mathsf{min}(A \union B) > 2 \cdot |h(\mathtt 0)| + 10$, then $h(\bar w_A) \equiv_{k} h( \bar w_B)$. 
\end{lem} 
\begin{proof}
Let $\Sigma$ be any alphabet where $| \Sigma | \geq 2$. 
Let $h \colon \{ \mathtt 0, \mathtt 1 \}^* \rightarrow \Sigma^*$ be a uniform morphism.
Let $\mathsf{m} \df |h(\mathtt 0)|$.
Assume that $\bar w_A \equiv_{k+20} \bar w_B$ for sets $A, B \subseteq \mathbb{N}$ where $\mathsf{min}(A \union B) > 2 \mathsf{m} + 10$.

We write $\mathcal{G}$ for the $k$-round game over $h(\bar w_A)$ and $h(\bar w_B)$.
Let $\mathcal{G}'$ be a $k+20$-round game over $\bar w_A$ and $\bar w_B$.
Using the fact that Duplicator plays $\mathcal{G}'$ using their winning strategy, we shall now show that $h(\bar w_A) \equiv_{k} h(\bar w_B)$.

\begin{figure}
    \centering
    \begin{tikzpicture}
    \tikzstyle{vertex}=[rectangle,fill=white!35,minimum size=12pt,inner sep=3pt,outer sep=1.5pt]
   \tikzstyle{vertex2}=[rectangle,fill=white!35,draw=gray,rounded corners=.1cm]
   
    \node[vertex] (1) at (-1.5, 1) {$u_i \cdot h(c_i) \cdot u_i'$};
    \node[vertex] (2) at (1.5, 1) {$c_i$};
    
    \node[vertex] (3) at (1.5,-0.5) {$d_i$};
    \node[vertex] (4) at (-1.5,-0.5) {$u_i \cdot h(d_i) \cdot u_i'$};
    
    \path[-latex,dashed] (1) edge node [above] {} (2);
    \path[-latex] (2) edge node [right] {$\mathcal{G}'$} (3);
    \path[-latex,dashed] (3) edge node [below] {} (4);
    \path[-latex] (1) edge node [left] {$\mathcal{G}$} (4);
    \end{tikzpicture}
    \caption{Duplicator's strategy when Spoiler chooses some $u$ where $|u| > 10 \mathsf{m}$.}
    \label{fig:dup-strategy}
\end{figure}

\paragraph*{Duplicator's Strategy}
Without loss of generality, let Spoiler chose some $a_i \sqsubseteq h(\bar w_A)$ in round $i$.
Now we define Duplicator's response:
\begin{enumerate}
\item If $|a_i| > 10 \mathsf{m}$, then let $(u_i , c_i , u_i')$ be the core factorization of $a_i$. Let Spoiler choose $c_i$ in round $i$ of $\mathcal{G}'$ and let $d_i$ be Duplicator's response in $\mathcal{G}'$. Then, Duplicator's response in $\mathcal{G}$ is $b_i \df u_i \cdot h(d_i) \cdot u_i'$. See~\autoref{fig:dup-strategy}.
\item If $|a_i| \leq 10 \mathsf{m}$ then let Duplicator respond with $b_i \df a_i$. We assume that Spoiler chooses $\emptyword$ in $\mathcal{G}'$ for round $i$; the structure Spoiler chooses here is not important.
\end{enumerate}

The case where Spoiler chose some $b_i \sqsubseteq h(\bar w_B)$ in round $i$ is defined symmetrically.

Let $\kappa = k+|\Sigma|+1$.
Let $\vec a = (a_1, \dots, a_{\kappa})$ and $\vec b = (b_1, \dots, b_{\kappa})$ be the resulting tuples from $\mathcal{G}$ where Duplicator plays their defined strategy.
Let $\vec c = (c_1, \dots, c_{\kappa})$ and $\vec d = (d_1, \dots, d_{\kappa})$ be the resulting tuples from the first $k$-round of $\mathcal{G}'$ where Duplicator plays their winning $k+20$-round strategy.

Note that Duplicator can still survive $\mathcal{G}'$ for an extra $20$ rounds, and therefore must play accordingly.
We shall utilize this fact to ensure Duplicator's strategy for $\mathcal{G}$ is indeed a winning strategy.
That is, we will look at possible choices for Spoiler in $\mathcal{G}'$ for rounds $k+1, \dots, k+20$ and show that since Duplicator must win $\mathcal{G}'$, this enforces Duplicator plays a way which translates to winning $\mathcal{G}$.
For a round $k+i$ of $\mathcal{G}'$ where $i \geq 1$, we denote the choice from $\bar w_A$ as $c_{\kappa + i}$ and the choice from $\bar w_B$ as $d_{\kappa + i}$.
As an example (which we use later on in this proof), from~\autoref{lemma:consistentStrats}, we know that if $|c_i| \leq 20$ for any $i \in [k]$, then $c_i = d_i$ must hold.
Likewise, if $|d_i| \leq 20$ for any $i \in [k]$, then $d_i = c_i$ must hold.

\paragraph*{Correctness}
In this section, we shall first prove that Duplicator's strategy is well-defined. 
Then, we shall show that it is indeed a winning strategy.

Clearly, for any choice Spoiler makes, there is a unique response defined for Duplicator, see~\autoref{lemma:unique-core}.
Thus, to show that Duplicator's strategy is well-defined, we must now show that Duplicator always responds with a factor of the corresponding word.
Assume without loss of generality, that in round $i$ of $\mathcal{G}$ Spoiler chooses $a_i \sqsubseteq h(\bar w_A)$.
We shall now prove that $b_i \sqsubseteq h(\bar w_B)$ does indeed hold, where $b_i$ is Duplicator's response as per their defined strategy.

First we show that if $a_i \leq 10 \mathsf{m}$ then $a_i \sqsubseteq h(\bar w_B)$.
Let $u \sqsubseteq \bar w_A$ where $|u|\leq 10$, then either:
$u = \mathtt{0}^r$ where $r \leq 10$, or
$u = \mathtt{0}^r \cdot \mathtt{1} \cdot \mathtt{0}^s$ with $r+s < 10$.
Since $\mathsf{min}(A) > 2 \mathsf{m} + 10$ and $\mathsf{min}(B) > 2 \mathsf{m} + 10$  we have that $u \sqsubseteq \bar w_B$.
Thus, for such an $a_i$, we can find some $u \sqsubseteq \bar w_A$ such that $a_i \sqsubseteq h(u)$ which implies $a_i \sqsubseteq h(\bar w_B)$.

Next consider $a_i > 10 \mathsf{m}$.
Then $( w_i , c_i , w_i')$ is the core factorization of $a_i$, where $c_i$ is the $i$-th element chosen from $\bar w_A$ in $\mathcal{G}'$.
Notice that by the definition of the core factorization, we have that $p_i \cdot c_i \cdot s_i \sqsubseteq \bar w_A$ where $p_i, s_i \in \{ \emptyword, \mathtt 0, \mathtt 1 \}$, and where $w_i \ssuff h(p_i)$ and $w_i' \ssuff h(s_i)$.
Moreover, $p_i = \emptyword$ if and only if $w_i = \emptyword$, and likewise $s_i = \emptyword$ if and only if $w_i' = \emptyword$.

Let $( w_i , d_i , w_i' )$ be the core factorization of $b_i$, where $d_i$ is the $i$-th element chosen from $\bar w_B$ in $\mathcal{G}'$. 
We claim that $p_i \cdot d_i \cdot s_i \sqsubseteq \bar w_B$.
To establish this claim, we show that assuming $p_i \cdot d_i \cdot s_i \sqsubseteq \bar w_B$ does not hold results in a contradiction.
Assuming $p_i \cdot d_i \cdot s_i \sqsubseteq \bar w_B$ does not hold, Spoiler can use the extra rounds of $\mathcal{G}'$ to choose the words $p_i \cdot c_i$, and $p_i \cdot c_i \cdot s_i$.
Since $p_i \cdot c_i$ is the concatenation of a terminal symbol (or the empty word) with a previous choice, Duplicator must respond with $p_i \cdot d_i$.
Analogously, Duplicator must respond to $p_i \cdot c_i \cdot s_i$ with $p_i \cdot d_i \cdot s_i$.
Since we know Duplicator plays $\mathcal{G}'$ using their winning strategy, $p_i \cdot d_i \cdot s_i \sqsubseteq \bar w_B$ and thus we obtain
$h(p_i ) \cdot h(d_i) \cdot h(s_i) \sqsubseteq h( \bar w_B)$.
It thus follows that $w_i \cdot h(d_i) \cdot w_i' \sqsubseteq h(\bar w_B)$.

We have proven that Duplicator's strategy is well defined.
Our next focus is to prove that Duplicator's strategy is a winning strategy.
The following claim is sufficient for showing Duplicator's strategy is a winning strategy.

\begin{clm}\label{claim:morphGenCorrect}
For all $l,i,j \in [\kappa]$ we have $a_l = a_i \cdot a_j$ if and only if $b_l = b_i \cdot b_j$.
\end{clm}

\noindent
The remainder of this proof establishes~\autoref{claim:morphGenCorrect}.
Assume, without loss of generality, that $a_l = a_i \cdot a_j$ for some $l,i,j \in [\kappa]$. 
We shall show that $b_l = b_i \cdot b_j$. 
The case where $b_l = b_i \cdot b_j$ implies $a_l = a_i \cdot a_j$ for $l,i,j \in [\kappa]$ follows symmetrically. 

To prove $a_l = a_i \cdot a_j$ implies $b_l = b_i \cdot b_j$, we use the following case distinction.
\begin{itemize}
\item \emph{Case 1:} $|a_i| > 10 \mathsf{m}$ and $|a_j| > 10 \mathsf{m}$.
\item \emph{Case 2:}  $|a_i| > 10 \mathsf{m}$ and $|a_j| \leq 10 \mathsf{m}$.
\item \emph{Case 3:} $|a_i| \leq 10 \mathsf{m}$ and $|a_j| > 10 \mathsf{m}$.
\item \emph{Case 4:} $|a_i| \leq 10 \mathsf{m}$ and $|a_j| \leq 10 \mathsf{m}$.
\end{itemize}

\emph{Case 1, $|a_i| > 10 \mathsf{m}$ and $|a_j| > 10 \mathsf{m}$:}
For $\gamma \in \{i,j, l\}$, let $a_\gamma = w_\gamma \cdot h(c_\gamma) \cdot w_\gamma'$ where $c_\gamma \sqsubseteq \bar w_A$ is the pre-image core of $a_\gamma$.
As $a_l = a_i \cdot a_j$, we have
\[ \underbrace{w_l \cdot h(c_l) \cdot w_l'}_{a_l} = \underbrace{w_i \cdot h(c_i) \cdot w_i'}_{a_i} \cdot \underbrace{w_j \cdot h(c_j) \cdot w_j'}_{a_j}. \]

We can take $w_l \cdot h(c_l) \cdot w_l'$ and factorize it as 
\begin{equation}\label{eq:factorize-al}
 a_l = w_l \cdot \underbrace{ h(c_{l,p}) \cdot w_{l,p} \cdot w_{l,s} \cdot h(c_{l,s}) }_{h(c_l)} \cdot w_l', 
 \end{equation}
where $w_{l,p} \cdot w_{l,s} = h(\lambda_l)$ for some $\lambda_l \in \{ \mathtt 0, \mathtt 1 \}$, and where $c_l = c_{l,p} \cdot \lambda_l \cdot c_{l,s}$.
Furthermore, with the appropriate factorization, it follows that $(w_l, h(c_{l,p}), w_{l,p}) $ is a core factorization of $a_i$. 
That is:
\[  a_l =  \underbrace{ w_l \cdot h(c_{l,p}) \cdot w_{l,p}}_{a_i} \cdot \underbrace{w_{l,s} \cdot h(c_{l,s})  \cdot w_l'. }_{a_j} \]
Since $|a_i| > 10 \mathsf{m}$, there is a unique core factorization, see~\autoref{lemma:unique-core}, and thus $w_l = w_i$, $c_{l,p} = c_i$, and $w_{l,p} = w_i'$.
Symmetrically, we know $w_{l,s} = w_j$, $c_{l,s} = c_j$, and $w_l' = w_j'$.

Now consider $b_i \cdot b_j$.
By the definiton of $\mathcal{G}'$ and Duplicator's strategy, we know that
\[ b_i \cdot b_j = w_i \cdot h(d_i) \cdot w_i' \cdot w_j \cdot h(d_j)  \cdot w_j'. \]
We also know that $w_i' \cdot w_j = h(\lambda_l)$, and therefore
\[ b_i \cdot b_j = w_i \cdot h(d_i) \cdot h(\lambda_l) \cdot h(d_j)  \cdot w_j'. \]
Likewise, we know that $w_i = w_l$ and $w_j' = w_l'$.
Hence
\[ b_i \cdot b_j = w_l \cdot h(d_i) \cdot h(\lambda_l) \cdot h(d_j)  \cdot w_l'. \]
As per the definition of the $\mathcal{G}'$ and Duplicator's strategy, we know that
\[ b_l = w_l \cdot h(d_l) \cdot w_l'. \]
Thus, if $d_l = d_i \cdot \lambda_l \cdot d_j$, then $b_l = b_i \cdot b_j$.
Note that $|h(\lambda_l) | \leq \mathsf{m}$, because $\lambda_l \in \{\mathtt 0, \mathtt 1\}$.

In round $k+1$, Spoiler can choose $h(\bar w_A)$ and $c_{\kappa+1} \df c_i \cdot \lambda_l$.
Hence, Duplicator must respond with $d_{\kappa+1} \df d_i \cdot \lambda_l$, since $\lambda_l$ is a constant.
Then, in round $k+2$, Spoiler can choose $c_{\kappa+2} \df c_{\kappa+1} \cdot c_j$ and Duplicator must respond with $d_{\kappa+2} \df d_{\kappa+1} \cdot d_j$.
As $c_{\kappa+2} = c_i \cdot \lambda_l \cdot c_j = c_l$, it must hold that $d_{\kappa+2} = d_i \cdot \lambda_l \cdot d_j = d_l$.
Thus, $d_l = d_i \cdot \lambda_l \cdot d_j$ and consequently $b_l = b_i \cdot b_j$.

\emph{Case 2,  $|a_i| > 10 \mathsf{m}$ and $|a_j| \leq 10 \mathsf{m}$:}
Note that (from the definition of Duplicator strategy) since $|a_j| \leq 10 \mathsf{m}$, we have that $b_j = a_j$.
Since $a_l = a_i \cdot a_j$, we can write
\[ w_l \cdot h(c_l) \cdot w_l' = w_i \cdot h(c_i) \cdot w_i' \cdot a_j,  \]
where $(w_l,c_l,w_l')$ is the unique core factorization of $a_l$ and $(w_i, c_i, w_i')$ is the unique core factorization of $a_i$.
	
Using the same reasoning as in Case 1 along with~\autoref{lemma:unique-core}, we have that $w_l = w_i$ and $c_i \pref c_l$.
Refer back to~\autoref{eq:factorize-al} for some intuition.
Thus, there is some $\lambda_l \sqsubseteq \bar w_A$ such that
\[ w_l \cdot \underbrace{h(c_i) \cdot h(\lambda_l)}_{h(c_l)} \cdot w_l' = w_l \cdot h(c_i) \cdot w_i' \cdot a_j. \]
Therefore,
$ h(\lambda_l) \cdot w_l' = w_i' \cdot a_j$.
Since $w_i' \leq \mathsf{m}$ and $a_j \leq  10 \mathsf{m}$, we know $|h(\lambda_l)| \leq 11 \mathsf{m}$, and hence $\lambda_l \leq 11$.
Spoiler can choose $c_{\kappa + 1} \df \lambda_l$ in round $k+1$ of $\mathcal{G}'$, and therefore Duplicator must respond with $d_{\kappa + 1} \df \lambda_l$, see~\autoref{lemma:consistentStrats}.
Furthermore, in round $k+2$ of $\mathcal{G}'$, Spoiler can choose $c_l = c_i \cdot \lambda_l$ and thus $d_l = d_i \cdot \lambda_l$ must hold.

Now let us consider 
\[ b_l = w_l \cdot h(d_l) \cdot w_l' . \]
As $d_l = d_i \cdot \lambda_l$, we have
\[ b_l = w_l \cdot h(d_i) \cdot h(\lambda_l) \cdot w_l' . \]

\noindent
Due to the fact that $ h(\lambda_l) \cdot w_l' = w_i' \cdot a_j$ and $b_j = a_j$, we know $h(\lambda_l) \cdot w_l' = w_i' \cdot b_j$.
Using this along with the fact that $w_l = w_i$, we get 
\[ b_l = w_i \cdot h(d_i) \cdot w_i' \cdot b_j. \]
Due to the definition of the $\mathcal{G}'$ and Duplicator's strategy, we know that $b_i = w_i \cdot h(d_i) \cdot w_i'$ and hence we arrive at $b_l = b_i \cdot b_j$.

\emph{Case 3,  $|a_i| \leq 10 \mathsf{m}$ and $|a_j| > 10 \mathsf{m}$:}
This follows symmetrically from Case 2. 

\emph{Case 4,$|a_i| \leq 10 \mathsf{m}$ and $|a_j| \leq 10 \mathsf{m}$:}
Note that since $|a_i| \leq 10 \mathsf{m}$ and $|a_j| \leq 10 \mathsf{m}$, we have that $a_i = b_i$ and $a_j = b_j$, by the definition of Duplicator's strategy.
If $|a_l| \leq 10 \mathsf{m}$ then it is trivial that $a_l = a_i \cdot a_j$ implies $b_l = b_i \cdot b_j$, as $a_l = b_l$.
Therefore, we continue this case under the assumption that $|a_l| > 10 \mathsf{m}$.
It follows that $a_l = w_l \cdot h(c_l) \cdot w_l'$ where $(w_l , c_l , w_l')$ is the core factorization, and $c_l$ is the element chosen from $\bar w_A$ in round $l$ of $\mathcal{G}'$.
Thus, $b_l = w_l \cdot h(d_l) \cdot w_l'$ where $d_l$ is the element chosen from $\bar w_B$ in round $l$ of $\mathcal{G}'$.
Since $|a_i| \leq 10 \mathsf{m}$ and $|a_j| \leq 10 \mathsf{m}$ it follows that $|a_l| \leq 20 \mathsf{m}$, and thus $c_l < 20$.
Thus, from~\autoref{lemma:consistentStrats} along with the fact that Duplicator must survive the extra $20$ rounds, we know that $d_l = c_l$ and hence $a_l = b_l$.
Consequently, $a_l = a_i \cdot a_j$ implies $b_l = b_i \cdot b_j$.

Thus, we have given a winning strategy for Duplicator for $\mathcal{G}$ using Duplicator's winning strategy for $\mathcal{G}'$.
This completes the proof of \autoref{lemma:morphGeneralize}.
\end{proof}

\autoref{lemma:morphGeneralize} is rather technical and specific to our purposes.
It gives us a tool to generalize Duplicator's strategy from~\autoref{lemma:evenOnes} using specific morphisms~$h \colon \{ \mathtt 0 , \mathtt 1 \}^* \rightarrow \Sigma^*$.
We shall see in the subsequent section that this is a key step in proving the language of a minimal DFA is not $\fc$-definable if it has a loop-step cycle.

\subsection{Loop-Step Cycles and Non-Definability in $\fc$}
This section gives the actual proof that if a minimal DFA $\mathcal{M}$ has a loop-step cycle, then $\lang(\mathcal{M}) \notin \lang(\fc)$.
We rely on the following result:

\begin{lemC}[\cite{thompson2023generalized}]\label{lemma:pseudoCongruence}
Let $w_1, w_2, v_1, v_2 \in \Sigma^*$ where $\facts(w_1) \intersect \facts(w_2)$ is equal to $\facts(v_1) \intersect \facts(v_2)$, and let $r$ be the length of the longest word in $\facts(w_1) \intersect \facts(w_2)$. 
If $w_1 \equiv_{k+ r + 2} v_1$ and $w_2 \equiv_{k+ r + 2} v_2$ for some $k \in \mathbb{N}_+$, then $w_1 \cdot w_2 \equiv_k v_1 \cdot v_2$.
\end{lemC}

\noindent
We are now ready for the main result of \autoref{sec:fc-gp}.

\begin{thm}\label{lemma:lynch_automata}\label{thm:fc-lsc}
Let $\mathcal{M} \df (Q, \Sigma, \delta, q_0, F)$ be a minimal DFA. 
If $\mathcal{M}$ has a loop-step cycle, then $\lang(\mathcal{M}) \notin \lang(\fc)$.  
\end{thm}
\begin{proof}
Let $\mathcal{M} \df (Q, \Sigma, \delta, q_0, F)$ be a minimal DFA
that has a loop-step cycle, i.e.,
there exist $n \geq 2$ pairwise distinct states $p_0, p_1,
\dots, p_{n-1}$ and words $w,v \in \Sigma^+$ where $\proot(w) \neq
\proot(v)$ and where: 
\begin{itemize}
\item $\delta^*(p_i,w) = p_i$ for all $i \in \{0, \dots, n-1 \}$,
\item $\delta^*(p_i, v) = p_{i+1 \pmod{n}}$.
\end{itemize}
Note that since $\mathcal{M}$ is minimal, all states must be
reachable, and thus there exists some word $p \in \Sigma^*$ such that $\delta^*(q_0, p) = p_0$.
Furthermore, there must exist some word $s \in \Sigma^*$ such that
$\delta^*(p_j, s) \in F$ for some $j \in \{0, \dots, n-1\}$ and
$\delta^*(p_\ell, s) \notin F$ for some $\ell \in \{0, \dots, n-1\}$
(the reason is that
otherwise, we could combine states to get a smaller automaton that
accepts the same language, contradicting the minimality of
$\mathcal{M}$; for details see Section~4.3 in Chapter~2 of~\cite{handbookOfFL}).

We choose such words $p$ and $s$ and indices
$0 \leq j,\ell < n$ and keep them fixed throughout the
remainder of this proof. That is, $\delta^*(q_0, p) = p_0$ and
$\delta^*(p_j, s) \in F$ and $\delta^*(p_\ell, s) \not\in F$.

Let $h \colon \{ \mathtt 0, \mathtt 1 \}^* \rightarrow \Sigma^*$ where:
\begin{itemize}
\item $h(\mathtt 0) \df w^{n  |v|} \cdot v^{n  |w| + n+1}$, and 
\item $h(\mathtt 1) \df w^{2n|v|} \cdot v^{n+1}$.
\end{itemize}
Observe that $|h(\mathtt 0)| = |h(\mathtt 1)|$, because
\[
|h(\mathtt 0)| = 2  n  |v|  |w| + n  |v| + |v|, \quad \text{and} \quad
|h(\mathtt 1)| = 2  n  |v| |w| + n |v| + |v|. 
\]
Notice that
$\delta^*(p_i, h(\mathtt 0)) = \delta^*(p_i, h(\mathtt 1)) = p_{i +1
  \pmod{n}}$
  for any $i \in \{ 0, \dots, n{-}1 \}$.

Theorem~16 of~\cite{lischke2011primitive} states that for two
primitive words $x,y \in \Sigma^+$ where $x \neq y$, the word $x^n
y^m$ is primitive for any $n,m \geq 2$.
Note that Theorem 16 of~\cite{lischke2011primitive} is a direct consequence of Theorem 9.2.4 of~\cite{lothaire1997combinatorics}.

From $\proot(w)\neq\proot(v)$ we thus obtain that
$h(\mathtt 0)$ and $h(\mathtt 1)$ are two distinct primitive words of equal length.
Thus, $h$ is a uniform morphism.
Therefore, we can apply \autoref{lemma:morphGeneralize} for any $k\in\NN$
and any choice of finite sets $A,B\subset\NN$ whose minimal elements
are sufficiently large.

Recall that $\delta^*(p_i,h(\mathtt 1))= p_{i+1 \pmod{n}}$ for every $i$ where $0 \leq i < n$. 
Furthermore, note that $\delta^*(p_i,h(\mathtt 0)^x)=p_i$
for every
$x\in\NN$ with $x\equiv 0 \pmod{n}$.

For any finite set $S\subset \NN$ such that $x\equiv
0\pmod{n}$ holds for all $x\in S$, consider the word
\[
 \alpha_S \df p \cdot h(\bar w_S) \cdot s.
\]
Note that the word $\bar w_S$ contains $|S|{+}1$ occurrences of the
letter $\mathtt{1}$, and between any two occurrences of the letter
$\mathtt{1}$, there are $x$ occurrences of the letter $\mathtt{0}$,
for some $x\in S$. Thus, for every $r\in\set{0,\ldots,n{-}1}$ the
following is true:
\begin{center}
If $|S|{+}1 \equiv r \pmod{n}$, then 
$\delta^*(p_0,h(\bar w_S))= p_r$.
\end{center}
Due to our particular choice of the words $s,p$ and the indices $j,\ell$,
we therefore obtain:
\begin{itemize}
\item
If $|S|{+}1\equiv j \pmod{n}$, then $\delta^*(q_0,\alpha_S)\in F$, and therefore
$\alpha_S\in \lang(\Aut)$.
\item
If $|S|{+}1\equiv \ell \pmod{n}$, then $\delta^*(q_0,\alpha_S)\not\in
F$, and therefore $\alpha_S\not\in\lang(\Aut)$.
\end{itemize}

\noindent
Recall that our goal is to show that $L\deff\lang(\Aut)$ is
not $\fc$-definable.
Due to \autoref{thm:ehrenfeucht} it suffices to construct for every
$k\in\NN$ two words $u\in L$ and $u'\not\in L$ such that $u\equiv_k u'$.

Let us fix an arbitrary $k\in\NN$, choose a suitable number $m$, and let $k'\deff r(k{+}m{+}20)$ be the
number provided by \autoref{lemma:evenOnes}.
Let $P_{k'} \subseteq \mathbb{N}$ be a $k'$-Lynchian set, which
satisfies the additional condition that $x\equiv 0 \pmod{n}$ for all $x\in P_{k'}$,
note that according to \autoref{defn:Lynchian}, this condition can
easily be satisfied.

From \autoref{thm:lynch} we know that $(\mathbb{N}, +, A) \equiv_{k'} (\mathbb{N}, +, B)$
for all finite sets $A, B \subseteq P_{k'}$ where $d(k') < |A|, |B|$.
As $k'$ is $r(k{+}m{+}20)$, by \autoref{lemma:evenOnes} we obtain: $\bar w_A \equiv_{k+m+20} \bar w_B$.
Provided that the minimal elements in $A$ and $B$ are sufficiently large,
\autoref{lemma:morphGeneralize} yields that $h(\bar w_A) \equiv_{k+m}
h(\bar w_B)$.
Using \autoref{lemma:pseudoCongruence}, and having chosen $m$
sufficiently large, we obtain that
$p\cdot h(\bar w_A)\cdot s \equiv_{k} p\cdot h(\bar w_B)\cdot s$,
i.e., we have $\alpha_A \equiv_k \alpha_B$.

Finally, we complete this proof by choosing two finite sets $A,B\subseteq P_{k'}$
where $d(k')<|A|,|B|$, and the minimal elements in $A$ and $B$ are large enough for
\autoref{lemma:morphGeneralize}, and where we have $|A|{+}1 \equiv j
\pmod{n}$ and $|B|{+}1\equiv \ell \pmod{n}$. Then, for the words
$u\deff \alpha_A$ and $u'\deff \alpha_B$ we have: $u\in L$, $u'\not\in
L$, and (provided that we have chosen $m$ large enough for applying \autoref{lemma:pseudoCongruence})  $u\equiv_k u'$. This proves that $L$ is not definable in $\fc$,
and it completes the proof of \autoref{thm:fc-lsc}.
\end{proof}

\noindent
The contraposition of~\autoref{thm:fc-lsc} states: $\lang(\mathcal{M}) \in \lang(\fc)$ implies $\mathcal{M}$ does \emph{not} have a loop-step cycle.
Thus, we have the final step of~\autoref{thm:main}, that being $\ref{item:mainthm:fc}\Rightarrow\ref{item:mainthm:lsc}$.
Consequently, this concludes the proof of~\autoref{thm:main}.

\section{Conclusions}\label{sec:conclusion}
In this paper, we have provided a decidable characterization of the $\fc$-definable regular languages in terms of (generalized) regular expressions, automata, and algebra (\autoref{thm:main}).
Moreover, we have shown that this characterization is decidable, and is in fact $\mathsf{PSPACE}$-complete for minimal DFAs (\autoref{thm:PSPACE}).

A promising next step would be to study the expressive power of fragments of $\fc$. 
The existential-positive fragment is particularly interesting due to its tight connection with core spanners~\cite{frey2019finite} and word equations (see Section~6.2 of~\cite{DBLP:journals/mst/Freydenberger19}).

\section*{Acknowledgements}
The authors would like to thank the anonymous reviewers for their valuable feedback and suggestions ---
in particular, regarding the comment that vastly simplified Section 7.4 of this work.
This work was supported by the EPSRC grant EP/T033762/1.
The first author would like to thank Joel Day for interesting discussions regarding this paper and related open problems.

\section*{Data Availability}
No data was analysed, captured or generated for this work.

\bibliographystyle{alphaurl}
\bibliography{ref}

\end{document}